\documentclass[journal]{IEEEtran}

\IEEEoverridecommandlockouts

\if CLASSOPTIONcompsoc
\usepackage[nocompress]{cite}
\else
\usepackage{cite}
\fi

\usepackage[utf8x]{inputenc}
\usepackage{listings}
\usepackage{color}
\usepackage{epsfig}
\usepackage{url}
\usepackage{amsmath}
\usepackage{graphicx}
\usepackage{color}
\usepackage{amssymb,amsmath,dsfont}
\usepackage{amsthm}
\usepackage{epstopdf}
\usepackage{lipsum}
\usepackage{bbm}
\usepackage{tablefootnote}

\usepackage{tikz}
\usetikzlibrary{automata,arrows,positioning,calc}

\usepackage{graphics}

\usepackage{csquotes}
\usepackage{amsmath} 
\usepackage{subcaption}
\usepackage{graphicx}
\usepackage{accents}
\usepackage{mathtools}
\usepackage[normalem]{ulem}
\usepackage{algorithm,algpseudocode}

\usepackage{multirow}

\usepackage{breqn}

\usepackage{fancyhdr}

\newtheorem{theorem}{Theorem}
\newtheorem{corollary}{Corollary}

\newtheorem{lemma}{Lemma}

\theoremstyle{remark}
\newtheorem*{remark}{Remark}
\newtheorem*{exmp}{Example}

\usepackage{tabularx}

\newcommand{\eqdef}{\vcentcolon=}

\graphicspath{figures/}

\usepackage{mathtools}

\DeclarePairedDelimiter\ceil{\lceil}{\rceil}

\newcommand{\ind}[1]{1_{#1}}
\newcommand{\calE}{\mathcal{E}}
\newcommand{\calL}{\mathcal{L}}
\newcommand{\calS}{\mathcal{S}}
\newcommand{\calX}{\mathcal{X}}
\newcommand{\Reals}{\mathbb{R}}

\newcommand\copyrighttext{%
	\footnotesize This work has been submitted to the IEEE for possible publication. Copyright may be transferred without notice, after which this version may no longer be accessible.
}
\newcommand\copyrightnotice{%
	\begin{tikzpicture}[remember picture,overlay]
	\node[anchor=north,yshift=-10pt] at (current page.north) {\parbox{\dimexpr\textwidth-\fboxsep-\fboxrule\relax}{\copyrighttext}};
	\end{tikzpicture}%
}

\begin{document}
	
\title{On Packet Reordering in Time-Sensitive Networks}
	
	
	\author{\IEEEauthorblockN{
					Ehsan Mohammadpour, Jean-Yves Le Boudec
			\\}
		\IEEEauthorblockA{\'Ecole Polytechnique F\'ed\'erale de Lausanne, Switzerland\\
			$\{$firstname.lastname$\}$@epfl.ch}}

	\maketitle
	
	\copyrightnotice

	\begin{abstract}

Time-sensitive networks (IEEE TSN or IETF DetNet) may tolerate some packet reordering. Re-sequencing buffers are then used to provide in-order delivery, the parameters of which (timeout, buffer size) may affect worst-case delay and delay jitter. There is so far no precise understanding of per-flow reordering metrics nor of the dimensioning of re-sequencing buffers in order to provide worst-case guarantees, as required in such networks. First, we show that a previously proposed per-flow metric, reordering late time offset (RTO), determines the timeout value. If the network is lossless, another previously defined metric, the reordering byte offset (RBO), determines the required buffer. If packet losses cannot be ignored, the required buffer may be larger than RBO, and depends on jitter, an arrival curve of the flow at its source, and the timeout. Then we develop a calculus to compute the RTO for a flow path; the method uses a novel relation with jitter and arrival curve, together with a decomposition of the path into non order-preserving and order-preserving elements. We also analyse the effect of re-sequencing buffers on worst-case delay, jitter and propagation of arrival curves. We show in particular that, in a lossless (but non order-preserving) network, re-sequencing is ``for free'', namely, it does not increase worst-case delay nor jitter, whereas in a lossy network, re-sequencing increases the worst-case delay and jitter. We apply the analysis to evaluate the performance impact of placing re-sequencing buffers at intermediate points and illustrate the results on two industrial test cases.
\end{abstract}

	\section{Introduction}\label{sec:intro}
Time-sensitive networks provide real-guarantees for applications in the automobile~\cite{ieeeDraftStandardLocal2019b}, automation~\cite{iecIECIEEE608022019}, space~\cite{ecssSpaceWireLinksNodes2008}, avionics~\cite{AFDX,TTE} and video~\cite{noauthor_ieee_2011-1} industries. Standardization is taking place at the IEEE Time Sensitive Networking (TSN) working group and at the IETF Deterministic Networking (DetNet) working group. In such networks, the aim is to provide flows with hard bounds on worst-case delay and on delay jitter (defined as the difference between worst-case and best-case delays), together with zero congestion loss and seamless redundancy~\cite{ieeeIEEEStandardLocal2017}.
	
Time-sensitive networks may allow some limited amount of
packet reordering. This may occur due to parallelism in network elements like switches and routers, routing of packets via different paths, or packet duplication \cite{bennett_packet_1999,laor_effect_2002,jaiswal_measurement_2007}. The IETF states in~\cite{rfc8655} that the amount of reordering is a key quality-of-service attribute of a flow; but neither IETF nor TSN specify what it means in detail. If a time-sensitive flow is subject to possible reordering and the application requires in-order packet delivery, a re-sequencing buffer is used to restore packet order. It is typically placed at the final destination, but it is also proposed in~\cite{rfc8655} to place at intermediate points inside the network, for example if the network path between the re-sequencing buffer and the destination preserves order, or simply to reduce the amount of reordering. A re-sequencing buffer uses the assumption that the source increments a sequence number field by 1 for every packet of the flow. Early packets are stored until all packets with smaller sequence numbers arrive \cite{piratla_metrics_2008,gao_large_2012,li_probabilistic_2010}. A timer is used to limit the waiting time of a packet in the re-sequencing buffer, as otherwise the loss of a packet in the network would cause indefinite holding of packets with larger sequence numbers.

Packet reordering is well understood in the context of best-effort networks where it is shown to be detrimental to the performance of TCP connections (see Section~\ref{sec:related}). Several metrics were proposed to capture the amount of reordering in \cite{piratla_metrics_2008} and in RFC4737~\cite{rfc4737}. In these references, the aim is to define reordering metrics that can be measured on a flow and can be correlated to the performance of reliable transfers over a best-effort network.

In time-sensitive networks, flows require a guarantee on worst-case delay and delay jitter, together with zero congestion-loss (no packet is discarded due to buffer overflow). To obtain such guarantees, a time-sensitive flow must conform to an arrival curve constraint at the source, which can be seen as a formal specification of a rate and burstiness constraint (see Section~\ref{sec:arrcur}). Then, using some forms of Network Calculus~\cite{le_boudec_network_2001}, the control or management plane computes worst-case delay and jitter bounds, together with the buffer sizes required for zero congestion-loss. Surprisingly, such computations are currently done without taking into account the impact of re-sequencing buffers. 

The main goal of this paper is to bridge this gap and provide a theory to compute worst-case performance guarantees in presence of packet reordering and with re-sequencing buffers.
Specifically, a first issue is to find appropriate per-flow information that enables proper setting of timeout value and buffer size at a re-sequencing buffer, such that deterministic guarantees hold for a time-sensitive flow, namely, no packet is lost due to spurious timeout or buffer overflow and packets are delivered in-order. 
A second issue is how to compute such per-flow information.
A third issue is the effect of the re-sequencing buffers on worst-case delay and on delay jitter. 
Furthermore, an intriguing topic is to study the aforementioned issues in the interconnection of various network elements and intermediate re-sequencing buffers. A first challenge is how to systematically compute the propagation of per-flow information in a sequence of network elements that affect the re-sequencing parameters. The second challenge is to evaluate the impact of intermediate re-sequencing buffers on the end-to-end worst-case delay and delay jitter.

Our contributions are as follows. To address the first and second issues, we need appropriate metrics for per-flow reordering. We show in Theorem~\ref{thm:reseq-param} that one of the metrics in RFC4737~\cite{rfc4737}, the reordering late time offset (RTO, the definition of which is recalled in Section~\ref{sec:metrics}), equals the minimal timeout value. Furthermore, combined with other information on the flow (namely arrival curve at source and delay jitter), the RTO can be used to derive the required buffer size (Theorem~\ref{thm:reseq-param-size}). In-line with the operation mode of time-sensitive networks, such a metric, or an upper bound on it, must be computed by the control or management plane before a flow is set up. This differs from the intended use of the metrics in RFC4737, which focus on ex-post measurements. Therefore, we propose a theory to compute tight upper bounds on RTO for flows, given the information that is otherwise available to the time-sensitive network control or management plane (Section~\ref{sec:reordering}). Such information includes bounds on delay jitter, arrival curve constraints of flows at their sources, and whether a network element is guaranteed to preserve per-flow order or not.

Another metric in RFC4737 is the reordering byte offset (RBO). We show in Theorem \ref{thm:reseq-param-size} that it is equal to the required size of the re-sequencing buffer when the network can be assumed to be lossless. Otherwise, if packet losses cannot be ignored, we show that the RBO underestimates the required buffer size, for which we give a formula that involves the RTO. This closes the first issue.


Concerning the third issue, observe that re-sequencing buffers may delay packets until they can be delivered in-order, therefore, they may increase the worst-case delay and the delay jitter. However, we show in Theorem~\ref{thm:delay-reseq} that, if a flow is lossless between its source and a re-sequencing buffer, the worst-case delay and the delay jitter are not increased by the re-sequencing buffer (i.e. re-sequencing is for free in terms of delay in the lossless case). In contrast, if the flow may be subject to packet losses on its path from source to the re-sequencing buffer, then the worst-case delay may be increased by an amount up to the timeout value of the re-sequencing buffer, which must be at least as large as the RTO between the source and the input of the re-sequencing buffer. These results are based on a novel input-output characterization of the re-sequencing buffer.


Finally, our theory also allows to evaluate the value of re-sequencing buffers at intermediate points, in addition to the destination, in time-sensitive networks. Regarding the first challenge, our formulas in Section \ref{sec:mco} capture in particular the pattern of RTO amplification by downstream jitter: if a non order-preserving element (typically a switching fabric) has very small RTO but is followed by a per-flow order-preserving element (typically the queuing system on an output port) with large delay jitter, then the concatenation of the two produces a large RTO. This motivates some vendors to perform per-flow re-sequencing after every switching fabric. With respect to the second challenge, we find that such intermediate re-sequencing buffers do not improve the worst-case delay nor delay jitter if the network is lossless; but they do reduce the worst-case delay, delay jitter and RTO at destination in presence of network losses. To quantify the effect of intermediate re-sequencing buffers, we also need to evaluate how arrival curves of flows are modified by re-sequencing, since such arrival curves are required to compute delay and jitter bounds (Section~\ref{sec:prp}). We illustrate the application of our theory to two industrial test cases.

The rest of the paper is organized as follows. The state-of-the-art is presented in Section~\ref{sec:related}. Common assumptions, including a formal description of 
the RTO and RBO metrics, are given in Section~\ref{sec:sys}, together with background results on network calculus in non-FIFO networks, and a notation list. In Section~\ref{sec:reseq-param} we provide a formal input-output characterization of the re-sequencing buffer, which is then used to establish the link between RTO and its required parameters, and to establish its effect on worst-case delay, delay jitter and output arrival curve. In Section~\ref{sec:reordering}, we show how RTO and RBO can be computed, as required to establish performance guarantees for time-sensitive flows; the method is in two parts: first, we develop formulas for an individual network element, given delay jitter and an arrival curve of the flow; then we develop a calculus to concatenate network elements.
In Section~\ref{sec:ir} we apply the results to analyze the performance of intermediate re-sequencing in two industrial case studies. Section~\ref{sec:conclusion} concludes the paper. Proofs of theorems and details of computations are in appendix. 
	\section{Related Work}\label{sec:related}

Kleinrock et al obtain the average re-sequencing delay in ~\cite{kamoun1982queueing}, assuming Poisson arrival of messages and a number of other simplifying assumptions. A more complete analysis is then performed in~\cite{baccelli1984end}, where the distribution of the end-to-end response time, including re-sequencing delay, is obtained.

Later studies mainly focus on the statistical measurement of the occurrence of reordering in a communication network; in \cite{paxson_end--end_1997,bennett_packet_1999}, the authors indicate that the rate of packet reordering is high inside the network.
 Later, other works focus on the real-time techniques to measure packet reordering \cite{measuring_bellardo_2002,wang_study_2004,piratla_reordering_2005}. In \cite{measuring_bellardo_2002}, the authors provide a collection of measurement techniques that can estimate end-to-end reordering rates in TCP connections. In \cite{wang_study_2004}, the authors propose and implement an algorithm to measure reordering at a TCP receiver. The authors in \cite{piratla_reordering_2005} provide the probability density function for the amount of reordering of an arbitrary packet, based on received packets.

All the aforementioned works focus on the techniques to capture statistical information on packet reordering inside the network. Few works study the sizing of re-sequencing buffers: \cite{li_probabilistic_2010,gao_large_2012} provide probability distribution of the re-sequencing buffer size. To the best of our knowledge, there is no prior work that computes the size of re-sequencing buffer and its timeout value in the context of worst-case performance (as required with time-sensitive networks) nor the effect of re-sequencing on worst-case delay and delay jitter.
	\section{Background Information}\label{sec:sys}
\subsection{Network Assumptions}
We consider a network that contains a set of nodes, a set of hosts, and a set of links with fixed capacity. Nodes are switches or routers. A node consists of elements that can be order-preserving for the flow of interest (e.g. output port FIFO queues) or non order-preserving (e.g. switching fabric). Every flow follows a fixed path, has a finite lifetime and emits a finite, but arbitrary, number of packets. We consider unicast flows (extension to single-source multicast flows is straightforward) with known arrival curves at their sources (i.e. there are known bounds on the number of bits or packets that can be emitted by a flow within any period of time, see Section \ref{sec:arrcur} for a formal definition of arrival curve). A node may also implement a re-sequencing buffer to provide in-order packet delivery for one or several flows of interest. If a flow requires in-order packet delivery and if there is at least one non-order-preserving element on its path, then one re-sequencing buffer is required, and can be placed anywhere after the last non order-preserving element on the path. In some configurations, we will also consider that some additional intermediate re-sequencing buffers are placed inside the network.

Hosts are sources or destinations of flows. Packet sequence numbers are written at the source, starting with number $1$ for the first packet sent by the flow. The sequence number is incremented by $1$ for every packet of the flow, i.e., sequence numbering is per-packet per-flow. If a packet is lost in the network, with most time-sensitive applications, there is no packet retransmission; instead, the application hides the loss using some application-specific robustness mechanism (see e.g. \cite{saab2018robust}). If the source happens to retransmit the missing data, the resulting packets are assumed to have a new sequence number (larger than the already sent packets of the same flow).

This section provides a set of general assumptions to present a high-level view of the considered network. Further details, e.g., scheduling policy inside the nodes, are not required for understanding the theory presented in this paper. In Section~\ref{sec:ir} we describe two case studies; there, we describe the network and flows with all details.


\subsection{Delay and Jitter}\label{sec:delay_jitter}

For a given flow, call $d_n$ the delay of the packet with sequence number $n$, measured from source to destination. The ``worst-case delay" of the flow is $\max_{n}\{d_n\}$ where the max is over all packets sent by the flow during its lifetime. Similarly, the ``best-case delay" of the flow is $\min_{n}\{d_n\}$. The ``delay jitter" is the difference, i.e.,
\begin{align}\label{eq:jitter-def}
V = \max_{n}\{d_n\} - \min_{m}\{d_m\},
\end{align} so that $d_m-d_n\leq V$ for any $m,n$. In the above, only packets that reach their destination should be considered (lost packets are excluded). Delay jitter is called IP Packet Delay Variation 
in RFC 3393 \cite{rfc3393}.

Times are assumed to be measured according to the true time, i.e.  the international atomic time (temps atomique international, TAI). In reality, times are measured with local clocks, which may or may not be synchronized. Some small corrections may need to be applied to delay and jitter bounds \cite{thomas2020time}; the details are for further study.

\subsection{Packet Reordering Metrics}\label{sec:metrics}
RFC 4737 defines a number of packet reordering metrics, two of which are of interest in the context of time-sensitive networks: the reordering late time offset (RTO) and the reordering byte offset (RBO), which we now formally define.
Both metrics are defined for a flow and between an input and an output observation points. When the input observation point is not specified, it is implicitly assumed that it is the source of the flow.

We call packet with index $n$ the $n$th packet observed at the input observation point, in chronological order. If the input observation point is the source, then packet~$n$ is the packet with sequence number $n$. Let $E_n$ be the time at which packet~$n$ is observed at the output observation point. If this packet is lost between the two observation points, we take $E_n=+\infty$. Note that, since some network elements may be non order-preserving, $E_n$ cannot be assumed to be a monotonic sequence.
Simultaneous packet observations might be possible in some cases (e.g. if the observation of a packet depends on the realization of a software condition) and the system must use some tie-breaking rule to determine a processing order for packets; if this happens, we assume that we modify the timestamps $E_n$ by some small amounts to reflect the tie-breaking rule i.e. we assume that if $j\neq n$ then $E_j  \neq E_n$.
%
If packet~$n$ is not lost, i.e. if $E_n<+\infty$, its reordering late time offset is
\begin{align}
\lambda_n = E_n - \min_{j | j \geq n, E_j \leq E_n}E_j
\end{align}
i.e. $\lambda_n$ is the largest amount of time by which a packet with index larger than $n$ arrives earlier than packet~$n$, i.e. the maximum amount of ``overtaking'' undergone by packet~$n$; if there is no reordering after packet~$n$, then $\lambda_n=0$. The reordering late time $\lambda_n$ is undefined if packet~$n$ is lost.

The reordering late time offset, RTO, of the flow between the two observation points is
$\lambda = \max_{n | E_n<+\infty}\lambda_n$.
It follows that, if packet~$n$ is not lost, then for any packet index $p \geq n$:
\begin{align}\label{eq:reordering-time-def}
E_p \geq  E_n - \lambda.
\end{align}
We always have $\lambda\geq 0$ and it is easy to see that $\lambda=0$ if and only if the network path between the two observation points preserves the order of packets for this flow.

For the second metric, we need to count misordered bytes; observe that 
packet $n$ is misordered between the two observation points if there exists some $j>n$ such that $E_j < E_n$.
Then, if packet~$n$ is not lost, its reordering byte offset is defined by
\begin{align}\label{eq:reordering-byte-def}
\pi_n = \sum_{j | j >n, E_j < E_n} l_{j}.
\end{align}
where $l_j$ is the size, in bytes, of packet~$j$. Thus $\pi_n$ is the cumulated number of bytes of packets with index larger than $n$ that arrive earlier than packet~$n$; if there is no reordered packet after $n$, then the sum is empty and $\pi_n=0$. The reordering byte offset $\pi_n$ is undefined if packet~$n$ is lost.
The reordering byte offset, RBO, of the flow between the two observation points is
$\pi = \max_{n | E_n<+\infty}\pi_n$.

This definition of RBO is in bytes and not in bits as is often done for buffer and packet lengths in the context of time-sensitive networks; this is to be consistent with the terminology in RFC~4737. Also observe that a similar definition could be given by counting packets rather than bytes, as is done in \cite{piratla_metrics_2008}.

\subsection{Re-sequencing Buffer}
A re-sequencing buffer stores the packets of a flow until the packets with smaller sequence numbers arrive; then it delivers them in the increasing order of their sequence numbers. A re-sequencing buffer has two parameters, a size in bytes, $B$, and a timeout value, $T$. For any individual packet that is stored in the buffer, a timer is set that expires after $T$ seconds. Then, if a timer for a packet expires, all the stored packets with smaller or equal sequence number are released in-order. Hence, a packet is released if any one of the following conditions holds: 1) all packets with smaller sequence numbers are received, 2) its timer expires or 3) the timer of a received packet with a larger sequence number expires. A detailed description of the re-sequencing buffer algorithm is provided in Appendix A.

By construction, the re-sequencing buffer delivers the packets that it does not discard in increasing sequence numbers. Furthermore, a packet is discarded by the re-sequencing buffer either when the buffer is full or when the sequence number of the arriving packet is less than the largest sequence number that was already released. The latter occurs when the timeouts of packets are too early, compared to the lateness of misordered packets. Therefore, to avoid discarding packets, the re-sequencing buffer size and the timeout value should be large enough. In Section \ref{sec:reseq-param}, we analyze how to set the parameters such that these conditions hold.

\subsection{Arrival Curves}\label{sec:arrcur}
In time-sensitive networks, in order to provide guarantees to flows in terms of delay, jitter and zero-congestion-loss, flow rates and burstinesses must be limited at the source. This is done precisely by imposing an ``arrival curve'' constraint, also called T-SPEC (traffic specification) at the source of the flow. Formally, let $\alpha$ be some wide-sense increasing function $[0,+\infty)\to [0,+\infty]$; the flow is said to satisfy the arrival curve constraint $\alpha$ at some observation point if the number of bytes (or bits) observed on the flow at this observation point on any interval $(s,t]$ is upper-bounded by $\alpha(t-s)$. Without loss of generality \cite{le_boudec_network_2001}, the arrival curve $\alpha$ can be assumed to be sub-additive ($\alpha(s+t)\leq\alpha(s)+\alpha(t)$ for all $s,t\geq 0$), left-continuous and such that $\alpha(0)=0$. In this paper we assume that entire packets are observed at the observation point, which imposes
that $\alpha(0^+)\geq L^{\max}$, where $L^{\max}$ is the maximal packet size, otherwise no packet of maximal size can be sent by the flow ($\alpha(0^+)$ stands for the right-limit of $\alpha$ at $0$).
An arrival curve $\alpha$ is ``achievable'' if, for any sequence of packet sizes between $L^{\min}$ and $L^{\max}$, there is a
source with these packet sizes that achieves equality in the arrival curve constraint, or more specifically, if
the sequence of packets obtained by packetizing the fluid source $R(t)=\alpha(t)$ satisfies the arrival curve constraint $\alpha$. Note that the fluid source $R(t)=\alpha(t)$ always satisfies the arrival curve constraint $\alpha$ since we can always assume that $\alpha$ is sub-additive; however packetization may introduce some violations \cite{le2002some}. Any concave arrival curve that satisfies $\alpha(0)=0$ and $\alpha(0^+)\geq L^{\max}$ is achievable \cite[Thm 1.7.3]{le_boudec_network_2001,le2002some}. For a flow with packets of constant size, any arrival curve whose values are integer multiples of the packet size is achievable as soon as it is sub-additive, left-continuous and satisfies $\alpha(0)=0$.


A commonly used arrival curve is the ``leaky bucket'' arrival curve with rate $r$ and burst $b$, defined by $\alpha(t)=rt+b, t>0$ and $\alpha(0)=0$,  which expresses that the rate of the flow is limited to $r$, with a burst tolerance equal to $b\geq L^{\max}$; it is always achievable. Another  commonly used arrival curve is the staircase arrival curve with period $\tau$ and burst $b$ defined by $\alpha(t)=b \left\lceil \frac{t}{\tau}\right\rceil$, which applies to periodic flows that send $b$ bytes every $\tau$ time units. It is achievable if all packets are of size $b$. If for example, in contrast, $b=1500$bytes but the packets emitted by the source all have a size equal to $1200$bytes, then this arrival curve is not achievable: $b$ should be set to $1200$, i.e. $\alpha$ should be replaced by a smaller arrival curve, which is then achievable.

Instead of counting bytes, constraints can be expressed in number of packets. Formally, let $\alpha_{\mathrm{pkt}}$ be some wide-sense increasing function $[0,+\infty)\to [0,+\infty]$; the flow is said to satisfy the arrival curve constraint $\alpha_{\mathrm{pkt}}$ at some observation point if the number of packets observed on the flow at this observation point on any interval $(s,t]$ is upper-bounded by $\alpha_{\mathrm{pkt}}(t-s)$. A packet-level arrival curve used by IEEE TSN is the staircase one, with period $\tau$ and number of packets $K$, defined by $\alpha_{\mathrm{pkt}}(t)=K \left\lceil \frac{t}{\tau}\right\rceil$, which expresses that the flow sends at most $K$ packets every $\tau$ time units. Packet-level arrival curves can always be replaced by an integer-valued sub-additive, left continuous function that vanishes at zero; it is then always achievable.

Last, we also need two technical assumptions. First, we assume that arrival curves are not bounded from above, i.e. $\lim_{t\to+\infty}\alpha(t)=+\infty$. This holds for all arrival curves of interest. Second, we assume that every flow has a maximum and minimum packet size $L^{\max}$ and $L^{\min}$; then the number of bytes observed on any time interval must be an element of $\calL$, the set of all possible sums of a finite number of packet sizes. If $L^{\max}\geq 2 L^{\min}$, then $\calL$ is made of all numbers $\geq L^{\min}$; if $L^{\max} =L^{\min}$ then $\calL$ is made of all multiples of $L^{\max} =L^{\min}$. Unless otherwise specified, we assume either of these conditions holds, as otherwise $\calL$ is cumbersome and tightness results would become very complex.
Appendix \ref{sec:appendix-arr} provides an alternative characterization of arrival curves that is used in the proofs of the theorems.

%
%
%
%
%

\subsection{Network Calculus Results in Non-FIFO Networks}
\label{sec-ncnf}
In a time-sensitive network, the burstiness of a flow may increase at every node, due to multiplexing and random delays. Thus an arrival curve constraint at the source is usually no longer valid inside the network. Analysis of time-sensitive networks uses bounds on the propagation of arrival curves \cite{bouillard_deterministic_2018}. Such results are based on network calculus theorems that were derived for order-preserving networks\cite[Section 1.4.1]{le_boudec_network_2001}, but which can be extended to non  order-preserving networks, as we show next. Specifically, we will use the following two results, the proofs of which are in appendix.
\begin{lemma}\label{lem:acp}
Assume a flow has arrival curve $\alpha$ at the input of some system $\calS$, which needs not preserve the order of packets of the flow.
Assume the delay jitter of the flow through $\calS$ is upper bounded by some quantity $V$. At the output of $\calS$, the flow has arrival curve $\alpha'$ given by $\alpha'(t)=\alpha(t+V)$.  The same holds, mutatis mutandi, for packet-level arrival curves.
\end{lemma}
\begin{lemma}\label{lem:backlog}
Assume a flow has arrival curve $\alpha$ at the input of some system $\calS$, which needs not preserve the order of packets of the flow.
Assume the worst-case delay of the flow through $\calS$ is upper bounded by some quantity $U$. At any point in time, the amount of data of the flow that is present in $\calS$ is upper-bounded by $\alpha(U)$.  The same holds, mutatis mutandi, for packet-level arrival curves.
\end{lemma}

\subsection{Lower Pseudo-inverse}
\label{sec:lps}
Let $F$ be a wide-sense increasing function $[0,+\infty)\to [0,+\infty)$. Its
lower pseudo-inverse, $F^{\downarrow}$, is the wide-sense increasing function $[0,+\infty)\to [0,+\infty)$ defined by \cite{liebeherr_duality_2017}:
\begin{align}\label{eq:def-lower-inverse}
F^{\downarrow}(x) = \inf\left\{s\geq 0~|~F(s)\geq x\right\}.
\end{align}
\begin{figure}[h]
	\centering
	\includegraphics[width=0.6 \linewidth]{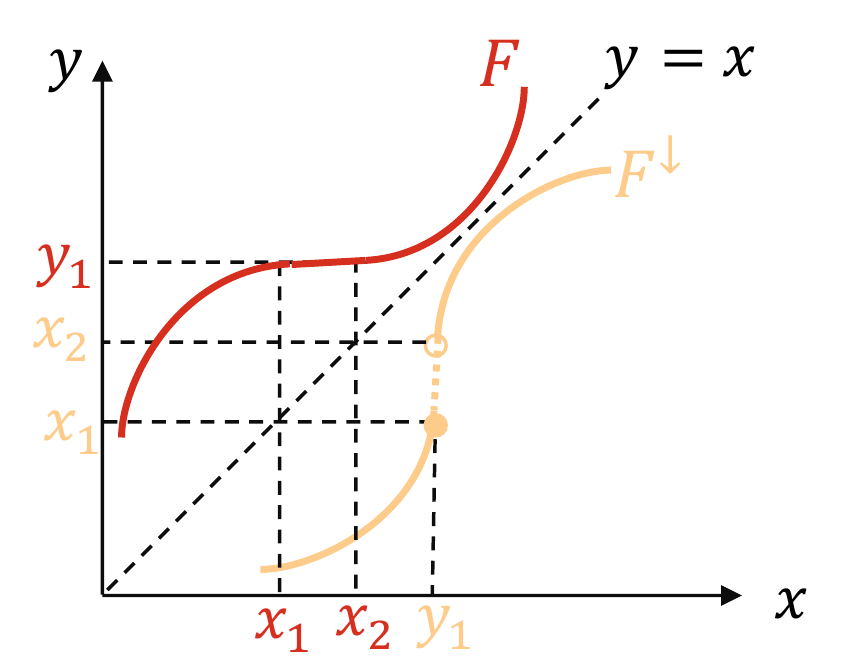}
	\caption{Illustration of lower pseudo-inverse of a monotonically increasing function $F$. The lower pseudo-inverse, $F^{\downarrow}$, is obtained by flipping the graph of $F$ around the line $y=x$. The resulting graph does not correspond to a function as the plateau part of $F$ i.e., $x_1$ to $x_2$, causes ambiguity. With the lower pseudo-inverse, the ambiguity is resolved by selecting the infimum,  i.e., $F^{\downarrow}(y_1)=x_1$.}
	\label{fig:pseudo}
\end{figure}
Some of the common functions and their lower-pseudo inverses are:
\begin{align} \label{eq:common-function-pseudo}
\nonumber F(t) = rt+b, t>0; F(0)=0 \indent \implies F^{\downarrow}(x) = \left[\frac{x-b}{r}\right]^+,\\
F(t) = b\ceil{\frac{t}{\tau}} \indent \implies F^{\downarrow}(x) = \tau \ceil{\frac{x-b}{b}}, x>0; F^{\downarrow}(0)=0.
\end{align}
It follows immediately from \cite[Property P7, Section 10.1]{liebeherr_duality_2017} that
\begin{align}\label{eq:gen-node-reordering-byte-11}
	\forall x,y \in [0,+\infty):~~F^{\downarrow}(y)  < x \implies y \leq F(x).
\end{align}


\subsection{Notation List}
Packet~$n$ is the packet with sequence number $n$ if the input point of observation is the source of the flow, otherwise it is the $n$th packet in chronological order at the input of the system of interest.
\begin{itemize}
	\item $l_n$: length of packet~$n$, in bytes.
	\item $L^{\min}$: minimum packet length of the flow of interest, in bytes.
	\item $L^{\max}$: maximum packet length of the flow of interest, in bytes.
	\item $A_n$: time at which packet~$n$ is released by its source.
	\item $D_n$: departure time of packet~$n$ from a re-sequencing buffer.
	\item $E_n$: the exit time of packet~$n$ from a non order-preserving or order-preserving element.
	\item $T$: timeout value of re-sequencing buffer.
	\item $B$: size of re-sequencing buffer.
	\item $\lambda$: reordering late time offset (RTO) of the flow of interest.
	\item $\pi$: reordering byte offset (RBO) of the flow of interest.
	\item the maximum of an empty set is $-\infty$.
	\item $[x]^+ = \max(x,0)$.
	\item $F^{\downarrow}$: lower pseudo-inverse of function $F$.
\end{itemize}

	\section{Properties of the Re-sequencing Buffer}\label{sec:reseq-param}
 In this Section we first provide a formal input-output characterization of the re-sequencing buffer. Then we use it to analyze the optimal parameter setting, assuming that bounds on RTO and RBO of the flow are known. Last, we characterize the performance effect of a re-sequencing buffer in terms of delay, delay jitter and arrival curve propagation.

\subsection{Input-output Characterization of the Re-sequencing Buffer}
We use the notation in Figure~\ref{fig:reseq-gen}. Recall that a packet may be lost in the network or discarded by the re-sequencing buffer. The latter may occur either (1) when the packet arrives after a packet with smaller sequence number was released; this is due to the timeout value $T$ being too small; or (2) if the buffer size $B$ is too small. In this subsection, we assume the buffer is large enough, and in Section~\ref{sec:cdsq} we compute the maximum buffer occupancy, which will give the required buffer size for a given flow. The following theorem characterizes the departure times from the re-sequencing buffer and is the basis from which the results in the rest of this section are derived.

\begin{figure}[h]
	\centering
	\includegraphics[width=0.8 \linewidth]{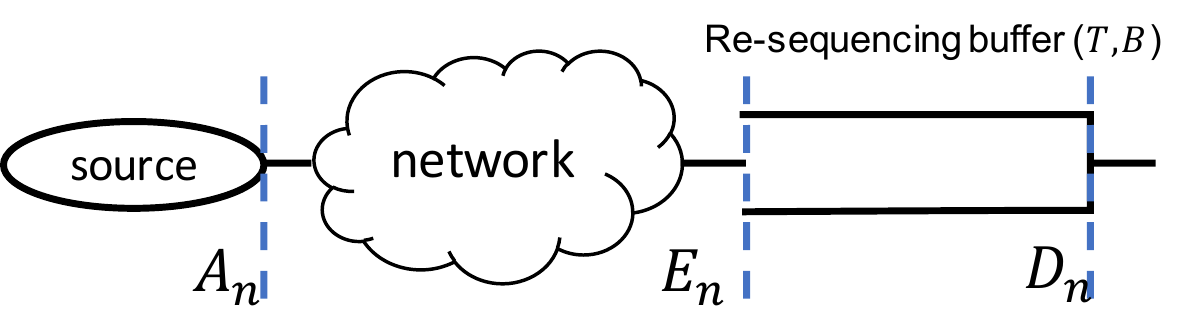}
	\caption{Notation used in Section~\ref{sec:reseq-param}. Packets of the flow of interest are emitted in sequence by a source. Packet with sequence number $n$ is emitted at time $A_n$, traverses a non order-preserving network, reaches the re-sequencing buffer at time $E_n$, from which it is released at time $D_n$. If the packet is lost by the network, then $E_n=D_n=+\infty$. If the re-sequencing buffer discards packet~$n$, then $D_n=+\infty$.}
	\label{fig:reseq-gen}
\end{figure}

\begin{theorem}
\label{thm:io}
  Consider the re-sequencing buffer described in Section \ref{sec:sys} with timeout value equal to $T$ and with infinite buffer capacity $B=+\infty$. See Figure~\ref{fig:reseq-gen} for the notation.
  \begin{enumerate}
    \item The packet with sequence number $n$ leaves the re-sequencing buffer at time $D_n$ given by
	\begin{align} \label{eq:reseq-io}
	D_n &= \begin{cases}
	I_n\indent&\mathrm{~if~}\indent n=1,\\
	\max\{G_n , I_n\} \indent &\mathrm{~if~} \indent n>1
	\end{cases}
	\end{align}
\begin{equation}
\label{eq:ioi}
	\mbox{with } I_n = \begin{cases}
	+\infty\indent\mathrm{~if~}\indent E_n>\min_{j\geq n}\{E_j\}+T,\\
	E_n \indent\mathrm{~otherwise}
	\end{cases}
\end{equation}
\begin{equation}
	 \mbox{and, for } n\geq 2: \;\;\;G_n = \min\left(D_{n-1}, T + \min_{j\geq n}\{E_j\}\right).\label{eq:iog}
\end{equation}	

    \item Let $\lambda$ be the RTO of this flow between the source and the input of the re-sequencing buffer. If $T\geq \lambda$ and packet~$n$ is not lost in the network (i.e. $E_n<+\infty$) then it also holds that
        \begin{align}
	D_n &= \begin{cases}
	E_1\indent&\mathrm{~if~}\indent n=1,\\
	\max\{G_n , E_n\} \indent &\mathrm{~if~} \indent n>1
	\end{cases}
    \label{eq:iod2}
	\end{align} where $G_n$ is defined in \eqref{eq:iog}.

    \item If $T\geq \lambda$ and the network is lossless (i.e. $E_n<+\infty$ for all $n$), then it also holds that
        $D_n = \max_{k\leq n}\left\{E_k\right\}$.
  \end{enumerate}

\end{theorem}
The proof Theorem \ref{thm:io} is in Appendix B. It is based on induction and uses the description of re-sequencing buffer presented in Section \ref{sec:sys}.

\subsection{Optimal Dimensioning of the Parameters of the Re-sequencing Buffer}
\label{sec:cdsq}
Recall that, by construction, the re-sequencing buffer always delivers packets in-order. However, it may do so by discarding late packets. We can now use the previous theorem to derive the minimal values of the timeout $T$ and the size $B$ of the re-sequencing buffer, such that it never discards any packet. We start with the timeout value.

\begin{theorem}\label{thm:reseq-param}
	Consider the re-sequencing buffer described in Section \ref{sec:sys} and Figure~\ref{fig:reseq-gen}, with timeout value of $T$ and infinite buffer size $B=+\infty$. Let $\lambda$ be the RTO of the flow of interest between the source and the input of the re-sequencing buffer. The minimum value of $T$ that guarantees that the re-sequencing buffer never discards packets of this flow is $T= \lambda$.
\end{theorem}
The proof of Theorem~\ref{thm:reseq-param} is in Appendix B.
 It consists in two steps. First, using Theorem~\ref{thm:io}, item~2, we show that, if $T\geq \lambda$, there is no packet discard due to spurious timeout. Second, using Theorem~\ref{thm:io}, item~1, we show that, for any $\lambda>0$, if $T<\lambda$ we can construct an execution trace with RTO $\lambda$ such that a packet is discarded due to spurious timeout.

Theorem~\ref{thm:reseq-param} thus establishes the central role of the RTO metric as far as the timeout value is concerned. For the required buffer size, the results are more complex, as shown in the next theorem.

\begin{theorem}\label{thm:reseq-param-size}
Consider the re-sequencing buffer described in Section \ref{sec:sys} and Fig.~\ref{fig:reseq-gen}, with timeout value of $T$ and buffer size $B$. Let $\lambda$, $\pi$ and $V$ be the RTO,
RBO and delay jitter of the flow of interest between the source and the input of the re-sequencing buffer. Assume that $T\geq \lambda$. Also assume that the flow has arrival curve $\alpha$ at its source. The minimal size of the re-sequencing buffer required to avoid buffer overflow is
	\begin{enumerate}
		\item  $B = \pi$, if the network in Fig.~\ref{fig:reseq-gen} is lossless for the flow;
		\item  $B = \alpha(V+T)$, if the network in Fig.~\ref{fig:reseq-gen} is not lossless for the flow.
	\end{enumerate}
\end{theorem}
The proof of Theorem~\ref{thm:reseq-param-size} is in Appendix B.
 It consists in four steps. First, assuming the network in Figure~\ref{fig:reseq-gen} is lossless for this flow and using Theorem~\ref{thm:io}, we show that the actual buffer content is upper bounded by $\pi$, which shows that a buffer of size $B=\pi$ is sufficient. Second, we show that for any $\lambda>0$ and any valid RBO value $\pi$ (a valid RBO value is a number that can be decomposed as the sum of an arbitrary number of packet sizes) there always exists one execution trace of a flow with RTO $\lambda$ and RBO $\pi$, that achieves a buffer content equal to $\pi$; therefore the minimal size cannot be less than $\pi$. If $\lambda=0$, the network preserves packet order for this flow and thus $\pi=0$ as well and the result is clear. This shows item~1.
Third, using Lemmas~\ref{lem:acp} and \ref{lem:backlog} in Section~\ref{sec:arrcur}, we show that, if the network is not lossless for this flow, the actual buffer content is upper bounded by $\alpha(V+T)$. 
Fourth, we show that, for any achievable arrival curve $\alpha$, RTO $\lambda$, jitter $V$ and timeout value $T$, we can construct an execution trace with RTO $\lambda$ in which the buffer content can become arbitrarily close to $\alpha(V+T)$.
This shows item~2.

\begin{remark}
It follows from
Theorem~\ref{thm:gen-node-reordering-byte} that the bound, $\pi$, in item~1, is always less than the bound, $\alpha(V+T)$, in item~2, as expected.
\end{remark}

\begin{remark}
Loss-free operation is often considered as the normal case in time-sensitive networks, since congestion losses are avoided and transmission losses are very rare; a packet loss might then seen as an exceptional error case, treated by exception-handling routines. If such an assumption can be made, Theorem~\ref{thm:reseq-param-size} shows that the required buffer content is only dependent on the RBO of the flow.

However, such an interpretation should be taken with care. Indeed, the loss of a single packet before the input to the re-sequencing buffer may delay a number of other packets: the first arriving packet with sequence number larger than the lost packet is delayed at the re-sequencing buffer by $T$, and, depending on the scenario, following packets may be delayed as well.
Thus, the loss of a single packet, even rarely, may impose a delay increase to many more subsequent packets and may lead to the violation of the bound in item 1.
Hence, if the bound in item~1 is used for dimensioning the re-sequencing buffer, then the loss of a single packet in the network may cause the loss of  many more packets at the re-sequencing buffer due to an insufficient buffer size (since the bound in item~2 is always larger than in item~1). Quantifying this in detail is left to further study.
\end{remark}

\begin{remark}
The RBO and buffer size $B$ in Theorem~\ref{thm:reseq-param-size} are  expressed in bytes. Obviously, a similar result holds if we count in packets: if the flow is constrained at the source by a packet-level arrival curve $\alpha_{\mathrm{pkt}}$, then the size of the re-sequencing buffer, counted in packets, is upper-bounded by $\alpha_{\mathrm{pkt}}(V+T)$.
\end{remark}

\subsection{Effect of Re-sequencing on Worst-case Delay, Jitter and Arrival Curve}
 \label{sec:prp}
When re-sequencing buffers are used, they may affect packet delay. In this section, we quantify this effect in the sense of worst-case delay and delay jitter, as required in time-sensitive networks.

\begin{theorem}\label{thm:delay-reseq}
	Consider a flow as in Figure~\ref{fig:reseq-gen}, and let $\lambda$ be the RTO of the flow between the source and the input of the re-sequencing buffer. Assume that the timeout value $T$ of the re-sequencing buffer satisfies $T\geq \lambda$.  The worst-case delay and the delay jitter of the flow
	\begin{enumerate}
		\item are not increased, if the network is lossless;
		\item are increased by up to $T$, if the network is not lossless.
	\end{enumerate}
\end{theorem}

Formally, with the notation in Figure~\ref{fig:reseq-gen}, the theorem means that, if the network is lossless (i.e. $E_n<+\infty$ for every $n$), then
\begin{align}
\max_n (D_n-A_n) = \max_n(E_n-A_n),
\end{align}
\begin{align}\label{eq:kljsda}
 \nonumber \max_n (D_n-A_n) - &\min_n (D_n-A_n)  = \\
  &\max_n(E_n-A_n)-\min_n(E_n-A_n)
\end{align}
as the former is the worst-case delay and the latter is the delay jitter.

In contrast, if there are some losses in the network, the theorem means that
\begin{align}
  \max_n (D_n-A_n) \leq \max_n(E_n-A_n)+T,
\end{align}
\begin{align}\label{eq:kljsda}
  \nonumber \max_n (D_n&-A_n) -\min_n (D_n-A_n)  \\
  &\leq \max_n(E_n-A_n)-\min_n(E_n-A_n)+T
\end{align}
The proof of Theorem~\ref{thm:delay-reseq} is in Appendix B.
It uses the input-output characterization of re-sequencing buffers in Theorem~\ref{thm:io}.

\begin{remark}
	Item (2) of Theorem \ref{thm:delay-reseq} is tight. Consider a packet~$n$ that experiences maximum delay $\delta^{\max}$ while packet~$n-1$ is lost in the network. Therefore, packet~$n$ should wait in the re-sequencing buffer until its timer expires after $T$ seconds. Then, packet~$n$ is delayed by $\delta^{\max} + T$.
\end{remark}

\begin{remark}
	Item (2) quantifies the price of misordering under lossy operation: the re-sequencing buffer, which is caused by the presence of misordering, increases the worst-case delay and the delay jitter of the flow by an amount ($T$) that is at least equal to the RTO.
\end{remark}

\begin{remark} The same remark about loss-free operation holds as in Section~\ref{sec:cdsq}. Specifically, the loss of a single packet may impose a delay increase to many more subsequent packets (e.g. to all packets that arrive before timeout). For example, if the flow has packet-level arrival curve $\alpha_{\mathrm{pkt}}$ at the source and jitter $V$ between the source and the input of the re-sequencing buffer, it can easily be seen, using the same arguments as in the proof of Theorem~\ref{thm:reseq-param-size}, that the loss of a single packet may cause the delay bound in item~1 to be violated for a number of packets equal to $\alpha_{\mathrm{pkt}}(V+T)$. This stresses again that the result in item~1 should be taken with care, and that the delay bounds in item~(2) are more realistic.
\end{remark}

We can apply Lemma~\ref{lem:acp} to the previous theorem and quantify the propagation of arrival curves through a re-sequencing buffer:
\begin{corollary}\label{coro:acp}
Consider a flow as in Figure~\ref{fig:reseq-gen}; assume that it satisfies the arrival curve $\alpha$ at the source and that the delay jitter between source and input to the re-sequencing buffer is $V$. Also assume that the timeout value $T$ of the re-sequencing buffer satisfies $T\geq \lambda$, where $\lambda$ is the RTO of the flow between the source and the input of the re-sequencing buffer. At the output of the re-sequencing buffer, the flow has arrival curve $\alpha'$ defined by
	\begin{enumerate}
		\item $\alpha'(t)=\alpha(t+V)$, if the network is lossless;
		\item $\alpha'(t)=\alpha(t+V+T)$, if the network is not lossless.
	\end{enumerate}
The same applies, mutatis mutandi, to packet-level arrival curves.
\end{corollary}

\begin{remark}
	Part (2) of the corollary is tight. This can be shown using the same arguments as in step (4) of the proof of Theorem~\ref{thm:reseq-param-size}. Specifically, the bound is achieved in a scenario where an isolated packet loss occurs, followed by a burst of in-order packets.
\end{remark}

	\section{Computing RTO and RBO}\label{sec:reordering}
In the previous section we saw how to dimension a re-sequencing buffer in the context of time-sensitive networks, assuming that we know the RTO of the flow and, to the extent that lossless metrics are of interest, its RBO. It remains to see how the RTO/RBO, or bounds on them, can be estimated by the control or management plane in order to setup the flow. To this end, we decompose a network path into elements that are either per-flow order-preserving or not. Examples of the former are the IEEE TSN class-based queuing subsystems \cite{mohammadpour_latency_2018}; examples of the latter are some switching fabrics which use parallel paths to improve throughput \cite{le2002packet}. In Section~\ref{sec:mne} we give tight RTO and RBO bounds for network elements; in Section~\ref{sec:mco} we show how to concatenate them.

\subsection{RTO and RBO for Network Elements}
\label{sec:mne}
\begin{theorem} \label{thm:gen-node-reordering-time}
Consider a flow that traverses a system, with delay jitter upper-bounded by $V$. Then an upper bound on the RTO of this flow between the input and the output of the system is:
	\begin{enumerate}
		\item $\left[V - \alpha^{\downarrow}(2L^{\min})\right]^+$, if the flow has arrival curve $\alpha$ at the input of the system;
		\item $\left[V - \alpha_{\mathrm{pkt}}^{\downarrow}(2)\right]^+$, if the flow has packet level arrival curve $\alpha_{\mathrm{pkt}}$  at the input of the system.
	\end{enumerate}
The bounds are tight, i.e. for every achievable arrival curve and every value of $L^{\min}$  and $V$ there is a system and an execution trace that attains the bound.
\end{theorem}
The proof is in Appendix B. It consists in two steps. First, based on the definition of jitter in Section \ref{sec:delay_jitter} and the alternative characterization of arrival curves in Appendix A, we show that the RTO is not larger than the bound in Theorem \ref{thm:gen-node-reordering-time}. Second, we show that there always exists one execution trace of a flow that its RTO is equal to the bound. Specifically, the bound is achieved in a scenario where the two packets enter a system in a greedy manner and the first one experiences the worse-case delay and the second one experiences the best-case delay.
 \begin{remark}
	If no arrival curve is known for the flow, we can always take $\alpha(t)=+\infty$ for $t>0$ and then $\alpha^{\downarrow}\left(x\right) =0$; this gives the RTO bound $\lambda=V$, i.e. jitter is a valid RTO bound for any system and any flow.
 \end{remark}
\begin{remark}
	If $\alpha^{\downarrow}(2L^{\min})\geq V$ or $\alpha_{\mathrm{pkt}}^{\downarrow}(2)\geq V$ then the RTO bound given by the theorem is $0$, i.e., there is no reordering for this flow. Thus, the theorem captures the cases where the packets sent by the flow are rare and the delay jitter of the non order-preserving system is small, so that reordering is impossible. See the next two examples.
 \end{remark}

\begin{remark}
	If the RTO of the flow is larger than zero, according to the theorem $\alpha^{\downarrow}(2L^{\min})< V$ or $\alpha_{\mathrm{pkt}}^{\downarrow}(2)< V$. Then by \eqref{eq:gen-node-reordering-byte-11}, we have $\alpha(V)\geq 2L^{\min}$ or $\alpha_{\mathrm{pkt}}(V) \geq 2$. This implies that if a flow has reordering, the input generates at least two packets within the duration of $V$.
\end{remark}

Hereafter, we provide examples on computation of RTO for multi-path connections and common forms of arrival curves, i.e., staircase and leaky bucket.

 \begin{exmp}
	Consider an interconnection system with $K$ paths. Every path $k$ has a worst-case delay $d^{\max}_k$ and a best-case delay $d^{\min}_k$. The delay jitter of this interconnection system is $V=\max_{k=1...K}  d^{\max}_k-\min_{k=1...K}  d^{\min}_k$; the RTO for a flow at the output of this interconnection is then given by Theorem~\ref{thm:gen-node-reordering-time}. Non order-preserving switching fabrics fall into this category; here, the delay jitter, and hence the RTO, are typically very small.
\end{exmp}

\begin{exmp}
	Consider a flow that has packet-level arrival curve $\alpha_{\mathrm{pkt}}(t) = K\ceil{\frac{t}{\tau}}$, expressing that at most $K$ packets are allowed in any time window of $\tau$ seconds. Here we have $\alpha_{\mathrm{pkt}}^{\downarrow}\left(x\right) = \tau \ceil{\frac{x-K}{K}}, x>0$ and Theorem~\ref{thm:gen-node-reordering-time} gives an RTO bound equal to $\lambda = \left[V - \tau \ceil{\frac{2-K}{K}}\right]^+$.

Applying this with $K=1$ gives that, for a flow that generates at most one packet every $\tau$ seconds:
	\begin{enumerate}
		\item if $\tau \geq V$, the flow experiences no reordering, i.e., $\lambda=0$.
		\item if $\tau < V$, the flow may experience reordering, and $\lambda = V-\tau$.
	\end{enumerate}
\end{exmp}

\begin{exmp}
	Consider a flow that has a leaky bucket arrival curve, i.e., $\alpha(t) = rt+b, t>0$, with rate $r$ and burst $b\geq L^{\max}$. Then, we have $\alpha^{\downarrow}\left(x\right) = \left[\frac{x-b}{r}\right]^+$ and the RTO bound given by Theorem~\ref{thm:gen-node-reordering-time} is $\lambda = \left[V - \left[\frac{2L^{\min}-b}{r}\right]^+\right]^+$. It follows that:
\begin{enumerate}
		\item if $b < 2L^{\min}$ and $V\leq\frac{2L^{\min}-b}{r}$, the flow experiences no reordering, i.e., $\lambda=0$;
		\item if $b < 2L^{\min}$ and $V>\frac{2L^{\min}-b}{r}$, the flow may experience reordering, and the RTO is bounded by $\lambda = V-\frac{2L^{\min}-b}{r}$;
       \item if $b \geq 2L^{\min}$ the flow may experience reordering, and the RTO is bounded by $\lambda = V$.
	\end{enumerate}

The first case requires $L^{\max}< 2L^{\min}$, which we excluded when $L^{\max}\geq 2L^{\min}$, but it may occur when packets are all of the same size $l$; then reordering is impossible if the delay jitter is $\leq \frac{2 l -b}{r}$.
\end{exmp}

\begin{theorem} \label{thm:gen-node-reordering-byte}
	Consider a flow that traverses a system, with delay jitter upper-bounded by $V$ and with RTO upper-bounded by $\lambda>0$. Then a bound on RBO of the flow between the input and the output of this system is:
	\begin{enumerate}
		\item $\alpha(V)- L^{\min}$, if the flow has arrival curve $\alpha$ at the input of the system and $\alpha(V)\geq 2L^{\min}$, and $0$ if $\alpha(V)< 2L^{\min}$;
		\item $L^{\max} \left(\alpha_{\mathrm{pkt}}(V)-1\right)$, if the flow has packet-level arrival curve $\alpha_{\mathrm{pkt}}$ at the input of the system and $\alpha_{\mathrm{pkt}}(V)\geq 2$, and $0$ if $\alpha_{\mathrm{pkt}}(V)< 2$.
	\end{enumerate}
The bounds are tight, i.e. for every achievable arrival curve and every value of $L^{\min}$, $L^{\max}$, $V$ and $\lambda>0$ 
there is a system and an execution trace such that the RBO of the flow is arbitrarily close to the bound.
\end{theorem}

The proof is in Appendix B. It consists in two steps. First,  based on the definition of jitter in Section \ref{sec:delay_jitter} and the alternative characterization of arrival curves in Appendix A, we show that the RBO is not larger than the bound in Theorem \ref{thm:gen-node-reordering-byte}. Second, we show that there always exists one execution trace of a flow that its RBO is arbitrary close to the bound. Specifically, it is achieved in a scenario where a sequence of packets enter a system in a greedy manner and the first packet experiences the worst-case delay while the other packets already left the system.

 Observe that we must have $\alpha(0^+)\geq  L^{\max}$ and $\alpha_{\mathrm{pkt}}(0^+) \geq 1$, therefore the expressions in items (1) and (2) are always non-negative. Also notice that the RBO bounds do not depend on $\lambda$ but require that $\lambda >0$; otherwise, namely if $\lambda=0$, there is no reordering and the RBO is $0$.
Last, observe that the tightness result implies that the RBO can be extremely large if the arrival curve can also be large. In other words, it is not possible to bound the RBO solely by constraining the delay jitter; for example, a non order-preserving switching fabric can have a very large RBO, limited only by the speed of the input ports, if the flows are not otherwise constrained.

\begin{exmp}
	Consider a flow that has a packet-level arrival curve $\alpha_{\mathrm{pkt}}(t) = K\ceil{\frac{t}{\tau}}$, expressing that $K$ packets are observed in any time window of $\tau$ seconds. An RBO bound for this flow is $\pi = K L^{\max}\ceil{\frac{V}{\tau}} - L^{\max}$ except if $K=1$ and $V\leq \tau$ in which case it is $0$.
\end{exmp}
\begin{exmp}
	Consider a flow that has a leaky bucket arrival curve, i.e., $\alpha(t) = rt+b, t>0$, with rate $r$ and burstiness $b\geq L^{\max}$. Applying Theorem \ref{thm:gen-node-reordering-byte} for a system with jitter bound $V$, an RBO bound for this flow is $\pi = rV+b-L^{\min}$.
\end{exmp}
\begin{remark}
If we count reordering offset in packets instead of bytes as in \cite{piratla_metrics_2008}, the bound in item (2) should be replaced by $\alpha_{\mathrm{pkt}}(V)-1$.
\end{remark} 
	\subsection{Concatenation Results}\label{sec:mco}

So far, we are able to compute RTO and RBO of a flow with known arrival curve for any system with known delay jitter. In practice, a flow typically traverses a sequence of network elements, of which some cause packet reordering and the rest preserve order. We are interested to compute RTO and RBO of the flow under such situation.
To tackle this problem, a trivial method is to concatenate all the elements as a single system with a delay jitter equal to the sum of delay jitters of each network element; then by applying theorems \ref{thm:gen-node-reordering-time} and \ref{thm:gen-node-reordering-byte}, we compute RBO and RTO of the combination.
However, we can do better, in two ways. First, as shown in Corollary \ref{col:reordering-concat}, we can ignore all prefix order-preserving elements when computing RTO and all suffix order-preserving elements when computing the RBO of the sequence. Second, we can use extra information about the RTO of every network element, as shown in Theorem \ref{thm:gen-concat-reordering-time}.

\begin{figure}[h]
	\centering
	\includegraphics[width= \linewidth]{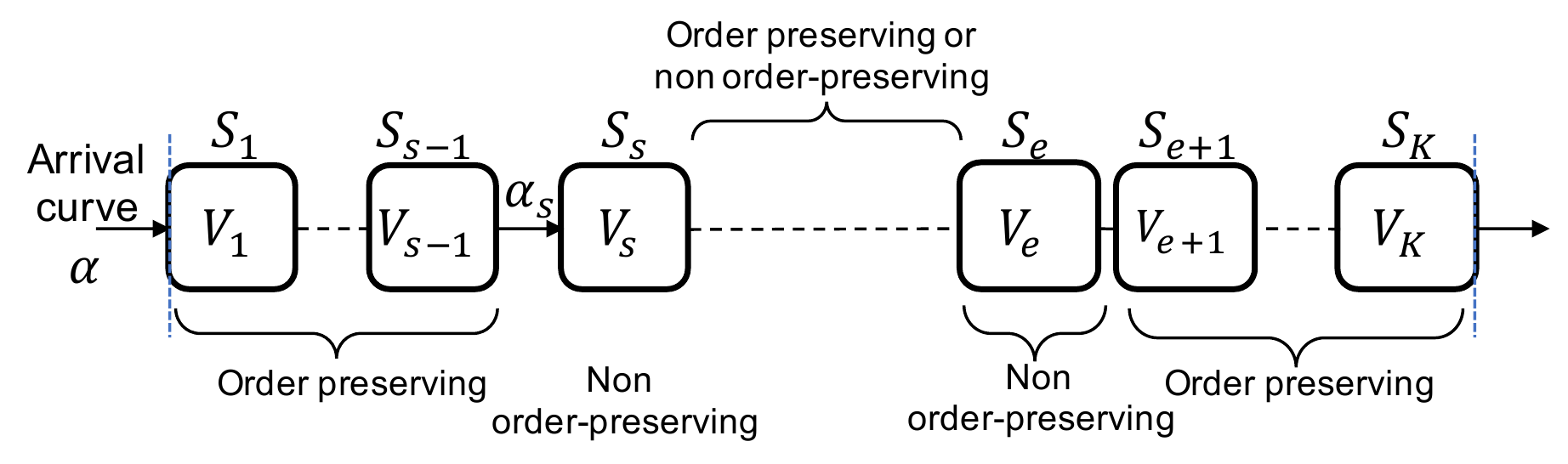}
	\caption{Notation for the sequence of network elements used in Section \ref{sec:mco}. $S_s$ and $S_e$ are respectively the first and the last non order-preserving elements in the sequence. $\alpha$ and $\alpha_s$ are arrival curves at the entrance of $S_1$ and $S_s$ respectively.}
	\label{fig:reordering-concat}
\end{figure}
\begin{corollary}\label{col:reordering-concat}
	Consider the notation in Fig. \ref{fig:reordering-concat}. An upper bound on the RTO of the flow between the input and the output of the sequence $S_1,...,S_K$, denoted as $\Lambda'(K)$, is obtained by applying Theorem \ref{thm:gen-node-reordering-time} only to $S_s,...,S_K$. Specifically:
	\begin{align}\label{eq:concat-th5-rto}
		\Lambda'(K)&= \sum_{h=s}^{K} V_h - \alpha_s^{\downarrow}(2L^{\min}),
	\end{align}
	where $\alpha_s$ is an arrival curve of the flow at the entrance of $S_s$.
\\
Moreover, an upper bound on the RBO of the flow between the input and the output of the sequence $S_1,...,S_K$ is obtained by applying Theorem \ref{thm:gen-node-reordering-byte} and setting the jitter term as $\sum_{h=1}^{e} V_h$ and the arrival curve term as the one at the entrance of $S_1$.
\end{corollary}
The proof of the Corollary is immediate: the RTO bound is justified by the definition of RTO and $S_1,...,S_{s-1}$ being order preserving. Similarly, since $S_{e+1},...,S_K$ are order preserving, they do not affect the RBO of the whole sequence.
In \eqref{eq:concat-th5-rto}, note that since $S_s$ is non order-preserving, $V_s - \alpha_s^{\downarrow}(2L^{\min}))>0$, thus $\sum_{h=s}^{K} V_h - \alpha_s^{\downarrow}(2L^{\min})>0$.

The bounds in Corollary \ref{col:reordering-concat} only exploit the information on the jitter of each element. However, as we now show, there are cases where an RTO bound $\lambda_h$ for each element $S_h$ can also be provided such that $\lambda_h < \lambda'_h$, where $\lambda'_h$ is obtained by Theorem \ref{thm:gen-node-reordering-time} using $V_h$ and an arrival curve at the entrance of $S_h$. Let us see the following example to see how this extra information can be available.
Consider a switch with different internal elements including order-preserving input processing units with jitter $V_1$, non order-preserving switching fabric with jitter $V_2$, and order-preserving output ports with jitter $V_3$. The jitter of the switch is $V=V_1+V_2+V_3$ and, a bound on the RTO of a flow with arrival curve $\alpha$ at the entrance of the switch is $\lambda=\left[ V_2+V_3- \alpha^{\downarrow}(2L^{\min}))\right]^+$ (Corollary~\ref{col:reordering-concat}). If only the jitter $V$ is exported by the switch, the RTO bound that can be computed by the control plane is $\lambda'=\left[ V- \alpha^{\downarrow}(2L^{\min}))\right]^+$ and $\lambda'>\lambda$ in most cases (i.e whenever $\alpha^{\downarrow}(2L^{\min})<V)$. Therefore, it is desirable that this switch exports both its jitter bound $V$ and its RTO bound $\lambda$.

This asks the question of which best RTO bound can be obtained from a concatenation of network elements, for each of which both jitter and RTO bounds are known. The answer is provided by the following theorem.

%
%
\begin{theorem}\label{thm:gen-concat-reordering-time}
Assume that for every network element $S_h$ and for the flow of interest, we know a bound $V_h$ on the delay jitter and a bound $\lambda_h$ on the RTO between the input and the output of $S_h$. Let $S_s$ be the first network element in the sequence that has $\lambda_s>0$. Then the RTO of the flow between the input and the output of the sequence is upper-bounded by
\begin{align}\label{eq:gen-concat-reordering-time}
	\Lambda(K) = \lambda_s+ \sum_{h=s+1}^{K} V_h.
\end{align}
The bound is tight, i.e., for every pair of sequences $V_h,\lambda_h$
there exists a system and an execution trace that comes arbitrarily close to the bound.
\end{theorem}

The proof is in Appendix B.
Obviously, we can assume that $V_s - \alpha_s^{\downarrow}(2L^{\min})>0$, since  $S_s$ is not order preserving for this flow, and, furthermore, that $\lambda_s\leq V_s - \alpha_s^{\downarrow}(2L^{\min})$, since otherwise, by Theorem~\ref{thm:gen-node-reordering-time}, we can replace $\lambda_s$ by $V_s - \alpha_s^{\downarrow}(2L^{\min})$.
It follows that the RTO bound given by Theorem~\ref{thm:gen-concat-reordering-time} is always at least as good as that of Theorem~\ref{thm:gen-node-reordering-time}, and improves on it whenever $\lambda_s< V_s - \alpha_s^{\downarrow}(2L^{\min})$, i.e. whenever the RTO bound exported by $S_s$ does improve over the knowledge of its jitter alone.

\begin{remark}
Theorem~\ref{thm:gen-concat-reordering-time} indicates that the only RTO that matters is that of the first non order-preserving element in the sequence (i.e. $S_s$). For subsequent network elements, it is their delay jitter $V_h$, not their RTO $\lambda_h$, that matters. Observe that RTO is always upper bounded by delay jitter (see remark after Theorem~\ref{thm:gen-node-reordering-time}). Therefore, the RTO of the flow through a sequence of nodes may be larger than the sum of the RTOs of the flow through each individual node. This can be explained as follows: If packet~$2$ overtakes packet~$1$ at node $S_s$ by a small amount up to $\lambda_s$, it may still happen that the delay of packet~$2$ through the subsequent nodes is much less than the delay of packet~$1$, by an amount up to the delay jitter $\sum_{h=s+1}^{K} V_h$. Thus, as destination, we observe that packet~$2$ overtakes packet~$1$ by an amount up to $\lambda_s+ \sum_{h=s+1}^{K} V_h$. Imagine that all non-order preserving network elements are switching fabrics, for which the RTO is tiny (sub-microseconds) and that other network elements are class-based queuing subsystems, which preserve packet order for every flow, but may have a much larger delay jitter (milliseconds or more). The end-to-end RTO is then of the order of milliseconds, orders of magnitude larger than the amount of reordering introduced by any single network element. This is the pattern of ``RTO amplification by downstream jitter''.
\end{remark}
Similarly, we can ask whether the information on RTO bound of every element can improve the RBO bound presented in Corollary \ref{col:reordering-concat}. The following theorem shows that the answer is no.
\begin{theorem}\label{thm:gen-reordering-byte-tight}
		The RBO bound in Corollary \ref{col:reordering-concat} is tight, i.e., for every pair of sequences $V_h,\lambda_h$ and every achievable arrival curve
		there exists a system and an execution trace that comes arbitrarily close to the bound.
\end{theorem}
The proof of tightness is similar to the proof of tightness in Theorem~\ref{thm:gen-node-reordering-byte}.

\section{Application to Performance of Intermediate Re-sequencing}\label{sec:ir}
In this section we illustrate how the results in the previous sections can be combined in order to evaluate the performance impact of intermediate re-sequencing.

\subsection{Methodology}
\label{sec:irm}

For a flow that requires in-order delivery but traverses a network where some elements do not preserve packet order, re-sequencing can be performed at the destination, but it is also possible to insert re-sequencing buffers at intermediate points. For example if one is placed for every flow at the output of every non order-preserving switching fabric, then the network becomes order-preserving and the end-system is relieved from the need to re-sequence. Every choice obviously comes with a different implementation cost; here we do not address such a cost. Instead, we focus on the performance impact, primarily in terms of end-to-end worst-case delay and delay jitter. In this illustration, we consider networks that do not perform flow re-shaping \footnote{Flow re-shaping refers to the process of recreating the arrival curve of a flow as its source.} \cite[Section 1.5]{le_boudec_network_2001}, inside the network.

If losses in the network are rare enough to be ignored for standard operation, the conclusion is straightforward. Indeed, we know from Theorem~\ref{thm:delay-reseq} that, under such an assumption, re-sequencing does not increase the worst-case delay and the delay jitter.

In contrast, if lossy operation cannot be ignored, we also know from Theorem~\ref{thm:delay-reseq} that re-sequencing adds a penalty to worst-case delay and jitter that is at least equal to the upstream RTO. Furthermore, the pattern of RTO amplification due to downstream jitter may mean that the RTO at the destination is very large, even though RTOs at non order-preserving elements are minuscule. This suggests that intermediate re-sequencing may be beneficial. However, intermediate re-sequencing also introduces a delay penalty and modifies the propagated arrival curves, which must be accounted for.
\begin{figure}[h]
	\centering
	\includegraphics[width= \linewidth]{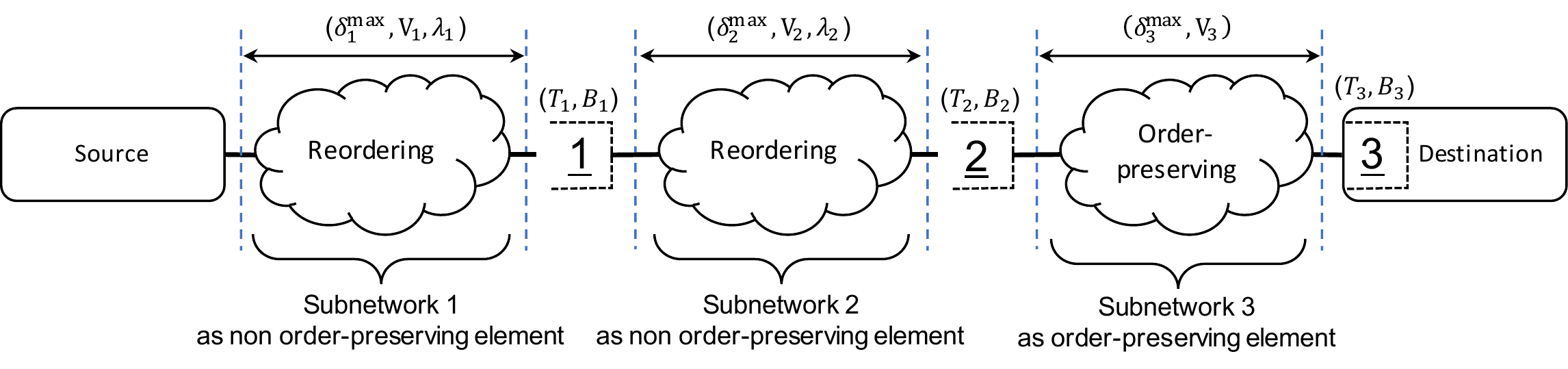}
	\caption{A prototypical scenario used to analyze the performance impact of intermediate re-sequencing. Potential placements of re-sequencing buffers are at points $1$, $2$, or $3$.}
	\label{fig:intermediate-reseq}
\end{figure}

To fix ideas, we consider a prototypical scenario as in Figure~\ref{fig:intermediate-reseq}, where the path of the flow of interest goes through three subnetworks, with jitters $\{V_i\}_{i=1}^3$ and RTOs $\{\lambda_i\}_{i=1}^3$. The first subnetwork may for example represent the source output queuing, a transmission link and the switching fabric of the next node. The second may represent the output queuing that follows this switching fabric, plus transmission links and the switching fabric of the following node. The third may represent the output queuing that follows the second switching fabric, plus transmission links to the final destination. By Theorem~\ref{thm:gen-concat-reordering-time}, we would have $\lambda_1 << V_1$ and $\lambda_2 << V_2$.
We consider the following possible placements of re-sequencing buffers:
\begin{itemize}
  \item Only at 3 (at destination end-system). 
  \item Only at 2 (at destination edge-switch). This is the case where the last edge-switch performs re-sequencing on behalf of the destination, just after the last non order-preserving element.
  \item At 1 and 2 (at every switch). This occurs when the network wants to guarantee that all switches preserve per-flow order, as some vendors do.
  \item At 1 and 3 (at the first switch and the destination end-system). This occurs when the network wants to guarantee that all switches except the destination edge-switch, preserve per-flow order. Then the destination performs re-sequencing with smaller timeout value.
\end{itemize}
We take as baseline the case where no re-sequencing is applied and compute the increase in worst-case delay and delay jitter with respect to the baseline, for each of the placements. We give the details for delay jitter, the computations are similar for the worst-case delay. We assume the timeout values are optimal, as given by Theorem~\ref{thm:gen-node-reordering-time}. The delay jitter of the baseline is $V_1+V_2+V_3$.

For the first placement (only at 3), the re-sequencing buffer at destination increases the delay jitter in the lossy case by $T_3$, the timeout value of the re-sequencing buffer at point 3, which is equal to the RTO between the source and point 3. By Theorem~\ref{thm:gen-concat-reordering-time}, it is equal to $\lambda_1+V_2+V_3$. By Theorem~\ref{thm:delay-reseq}, the increase on jitter is also $\lambda_1+V_2+V_3$.

\begin{table*}[t]
	\caption{Re-sequencing buffer optimal timeout value and increase on end-to-end jitter and delay upper bound (with respect to the baseline with no re-sequencing buffer) for the four placement strategies in Section~\ref{sec:irm}. $T_i$ is the timeout of re-sequencing buffer placed at point $i$. We see that placing re-sequencing buffers at 1 and 2 provides better end-to-end delay and jitter comparing to the placement at 1 and 3.}
\label{tab:conc}
\resizebox{\textwidth}{!}{
	\begin{tabular}{|c|c|c|c|c|c|}
		\hline
		\multirow{3}{*}{Re-sequencing} & \multicolumn{3}{c|}{Timeout} & \multicolumn{2}{c|}{\multirow{2}{*}{\begin{tabular}[c]{@{}c@{}}Increase to end-to-end jitter and delay with \\ respect to no re-sequencing buffer at all.\end{tabular}}} \\ \cline{2-4}
		& $T_1$ & $T_2$ & $T_3$ & \multicolumn{2}{c|}{} \\ \cline{2-6}
		& \multicolumn{3}{c|}{Lossless and lossy} & Lossless & Lossy \\ \hline
		Only at 3 & - & - & $\lambda_1+V_2+V_3$ & $0$ & $\lambda_1+V_2+V_3$ \\
\hline
		Only at 2 & - & $\lambda_1+V_2$ & - & $0$ & $\lambda_1+V_2+\Delta V_3$ \\
\hline
		At 1 and 2 & $\lambda_1$ & $\lambda_2$ & - & $0$ & $\lambda_1+\lambda_2+\Delta V_2+\Delta V_3$ \\
\hline
At 1 and 3 & $\lambda_1$ & - & $\lambda_2+V_3(+\Delta V_3  \mbox{ for lossy})$ & $0$ & $\lambda_1+\lambda_2+V_3+\Delta V_2 + 2\Delta V_3$ \\ \hline
	\end{tabular}
}
\end{table*}

For the last placement (at 1 and 3), the re-sequencing buffer at 1 modifies the arrival curve of the flow. This affects, in general, the downstream worst-case delay and jitter. We call $\Delta V_2$ and $\Delta V_3$ the increase on delay jitter at subnetworks 2 and 3 with respect to the baseline. In the following subsection we estimate these increases numerically on two industrial cases. The re-sequencing buffer at 1 has timeout $T_1=\lambda_1$, given by the RTO of subnetwork 1. For the RTO at point 3, observe that
the subnetwork 1 combined with the re-sequencing buffer at point 1 is an order-preserving element; this causes the RTO at point 3 to be independent of subnetwork 1 (Theorem \ref{thm:gen-concat-reordering-time}). We obtain the timeout value $T_3=\lambda_2+ V_3 +\Delta V_3$. Finally, the end-to-end delay jitter is increased by $T_1$ at point 1 and $T_3$ at point 3, thus it is equal to $V_1+T_1+V_2+\Delta V_2+ V_3+\Delta V_3+ T_3$, which gives an increase with respect to the baseline equal to $\lambda_1+\lambda_2+V_3+\Delta V_2 + 2\Delta V_3$.

The reasoning is similar for the two other placements. The results are given in Table~\ref{tab:conc}. A similar line of reasoning can be used to compute the required sizes of re-sequencing buffers.

We observe the following. In this scenario, we expect $\lambda_i$ to be much smaller than $V_i$, so if $\Delta V_i$ is small, it is beneficial to place an intermediate re-sequencing buffer, and it is also beneficial to place one at the edge-node rather than at the destination. This is because intermediate re-sequencing reduces the downstream RTO and avoids the RTO amplification pattern. However, if $\Delta V_i$ is large this benefit may be lost due to the burstiness increase caused by re-sequencing under lossy operation. In the numerical examples of the next sections, we find that, except in one case, the former effect is largely dominant. Also note that if per-flow re-shaping would be performed at every hop, the latter effect would disappear and intermediate re-sequencing would always reduce the worst-case delay and jitter under lossy operation.

\subsection{Case Study 1: Automotive Network}\label{subsec:case1}

We apply the methodology in Section~\ref{sec:irm} to the double star automotive network of \cite{manderscheid_network_2011} depicted in Fig. \ref{fig:auto}.
To obtain the re-sequencing buffer size and timeout, as well as the jitter and delay upper bounds, we used TFA \cite{thomas_on-cyclic_2019,mifdaoui_beyond_2017}  (details are in Appendix~C).
\begin{figure}[h]
	\centering
	\includegraphics[width= \linewidth]{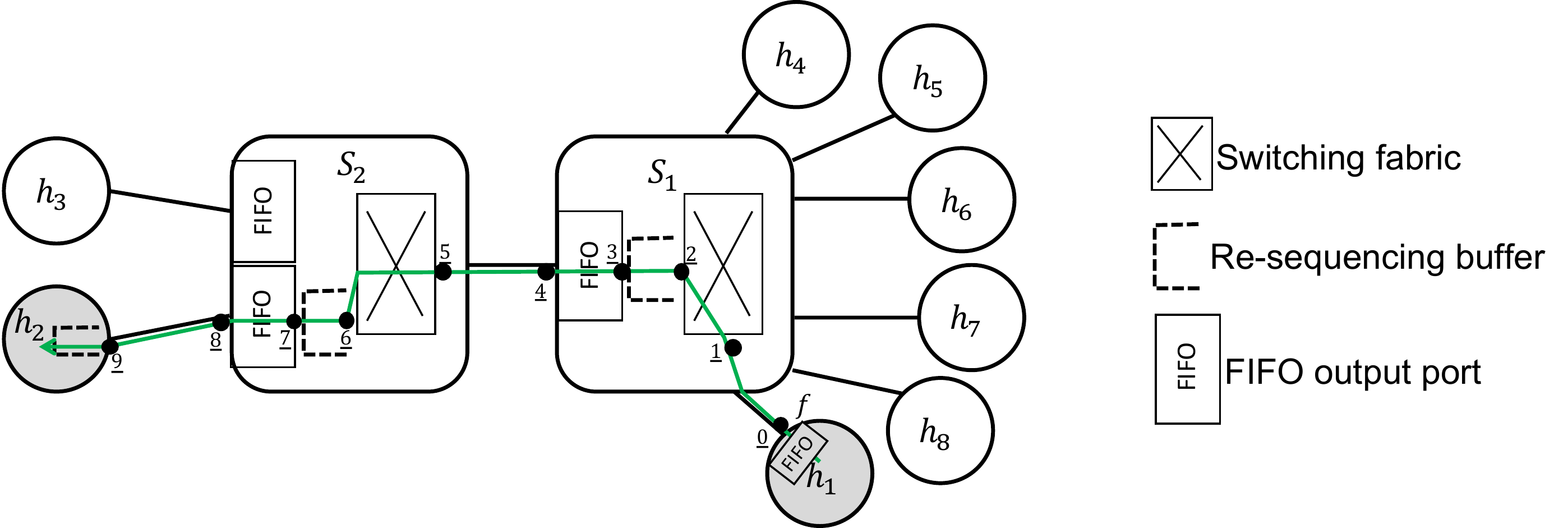}
	\caption{Double star automotive network \cite{manderscheid_network_2011}.}
	\label{fig:auto}
\end{figure}

The network consists of two switches and eight hosts and the rates of the links are $c = 1$~Gbps.
 The output ports are FIFO and the scheduling mechanism is non-preemptive strict priority. The output ports of $h_1$, $S_1$ and $S_2$ offer the same service curve $\beta(t) = 125e6 [t-12\mu]^+$~bytes to the highest priority queue.
 In each switch, the switching fabric is implemented in parallel stages, i.e., reordering of packets may occur. The delay of switching fabrics is between $0.5 \mu s$ to $2 \mu s$ \cite{nexus9508}.

According to \cite{manderscheid_network_2011}, the traffic is made of various flows with different priorities. In our example, we focus on ControlData flow as the only highest priority flow; it is shown as flow $f$ in Fig. \ref{fig:auto} and the path is taken from \cite{manderscheid_network_2011}; the flow is initiated by Control Data Unit ($h_1$) and destined to Control Unit ($h_2$) in the network. The source arrival curve for flow $f$ is leaky bucket with rate $6400$ bytes per second and burstiness $6400$~bytes. All packets are the same size and equal to $64$~bytes.

We implemented the four placement strategies in Section~\ref{sec:ir}.
The results are in Table~\ref{t:uc1}. We see that re-sequencing at every switch fabric significantly reduces delay and jitter bounds. In contrast, re-sequencing at edge node ($S_2$ only) is not beneficial: this is an instance where the burstiness increase due to re-sequencing does have an impact. We also see that the required size of the re-sequencing buffer is independent of the placement strategy.

\begin{table*}[t]
	\centering
	\caption{Bounds on end-to-end jitter and worst-case delay for case study 1 with the four placement strategies in Section~\ref{sec:irm}, under both lossless and lossy network conditions, followed by timeout value and size of the re-sequencing buffers.}
	\resizebox{\textwidth}{!}{
\begin{tabular}{|c||c|c||c|c|}
		\hline
		\multirow{2}{*}{\begin{tabular}[c]{@{}c@{}}Re-sequencing buffers\\ placement\end{tabular}} & \multicolumn{2}{c||}{Lossless} & \multicolumn{2}{c|}{Lossy} \\ \cline{2-5}
		& Delay ($\mu s$) & Jitter ($\mu s$) & Delay ($\mu s$) & Jitter ($\mu s$) \\ \hline\hline
		Only at $h_2$      & $95.22$         & $92.69$          & $124.72$        & $122.19$         \\ \hline
		Only at $S_2$      & $95.22$         & $92.69$          & $127.22$        & $124.69$         \\ \hline
		At $S_1$ and $h_2$ & $95.22$         & $92.69$          & $111.72$        & $109.19$         \\ \hline
		At $S_1$ and $S_2$ & $95.22$         & $92.69$          & $99.22$         & $96.69$          \\ \hline
	\end{tabular}

\begin{tabular}{|c||c|c|c||c|c|c||c|c|c|}
		\hline
		\multirow{3}{*}{\begin{tabular}[c]{@{}c@{}}Re-sequencing\\ buffers\\ placement\end{tabular}} &
		\multicolumn{3}{c||}{Lossless and Lossy} &
		\multicolumn{3}{c||}{Lossless} &
		\multicolumn{3}{c|}{Lossy} \\ \cline{2-10}
		& \multicolumn{3}{c||}{Timeout $T$ ($\mu s)$} & \multicolumn{3}{c||}{Size $B$ (bytes)} & \multicolumn{3}{c|}{Size $B$ (bytes)} \\ \cline{2-10}
		& $S_1$       & $S_2$       & $h_2$       & $S_2$      & $S_2$     & $h_2$    & $S_1$     & $S_2$     & $h_2$     \\ \hline \hline
		Only at $h_2$      & $-$         & $-$         & $29.49$     & $-$        & $-$   & $6336$   & $-$       & $-$       & $6400$    \\ \hline
		Only at $S_2$      & $-$         & $15.99$     & $-$         & $-$    & $6336$       & $-$      & $-$       & $6400$    & $-$       \\ \hline
		At $S_1$ and $h_2$ & $0.98$      & $-$         & $14.49$     & $6336$        & $-$   & $6336$     & $6400$    & $-$       & $6400$    \\ \hline
		At $S_1$ and $S_2$ & $0.98$      & $0.98$      & $-$         & $6336$     & $6336$       & $-$      & $6400$    & $6400$    & $-$       \\ \hline
	\end{tabular}
}
	\label{t:uc1}
\end{table*}

\subsection{Case study 2: Orion network}

We now consider the Orion crew exploration vehicle network, as described in \cite{obermaisser_time-triggered_2012} and depicted in Fig. \ref{fig:orion}, taken from \cite{thomas_on-cyclic_2019}. For the delay and jitter analysis, we used Fixed Point TFA \cite{thomas_on-cyclic_2019,mifdaoui_beyond_2017} as there are cyclic dependencies in the placement of flows. The output ports in the hosts and switches are connected to the links with a rate of $1$~Gbps. The output ports use the non-preemptive TSN scheduler with Credit-based Shapers (CBSs) with per-class queuing \cite{mohammadpour_latency_2018,zhao_timing_2018}; from highest to lowest priority, the classes are Control Data Traffic (CDT), A, B, and Best Effort). The CBSs are used separately for classes A and B. The CBS parameters $\mathit{idleslopes}$ are set to $50\%$ and $25\%$ of the link rate respectively for classes A and B \cite{zhao_timing_2018}. In each switch, the switching fabric is implemented in parallel stages, i.e., reordering of packets may occur. The delay of switching fabrics is between $0.5 \mu s$ to $2 \mu s$ \cite{nexus9508}. The CDT traffic has a leaky bucket arrival curve with rate $6.4$ kilobytes per second and burst $64$~bytes. The maximum packet length of classes B and BE is $1500$~bytes. We focus on class A. Using the results in \cite{mohammadpour_latency_2018}, a rate-latency service curve offered to class A is $\beta(t) = 62.49e6[t-t_0]^+$~bytes with $t_0=12.5\mu$s.

\begin{figure}[h]
	\centering
	\includegraphics[width=\linewidth]{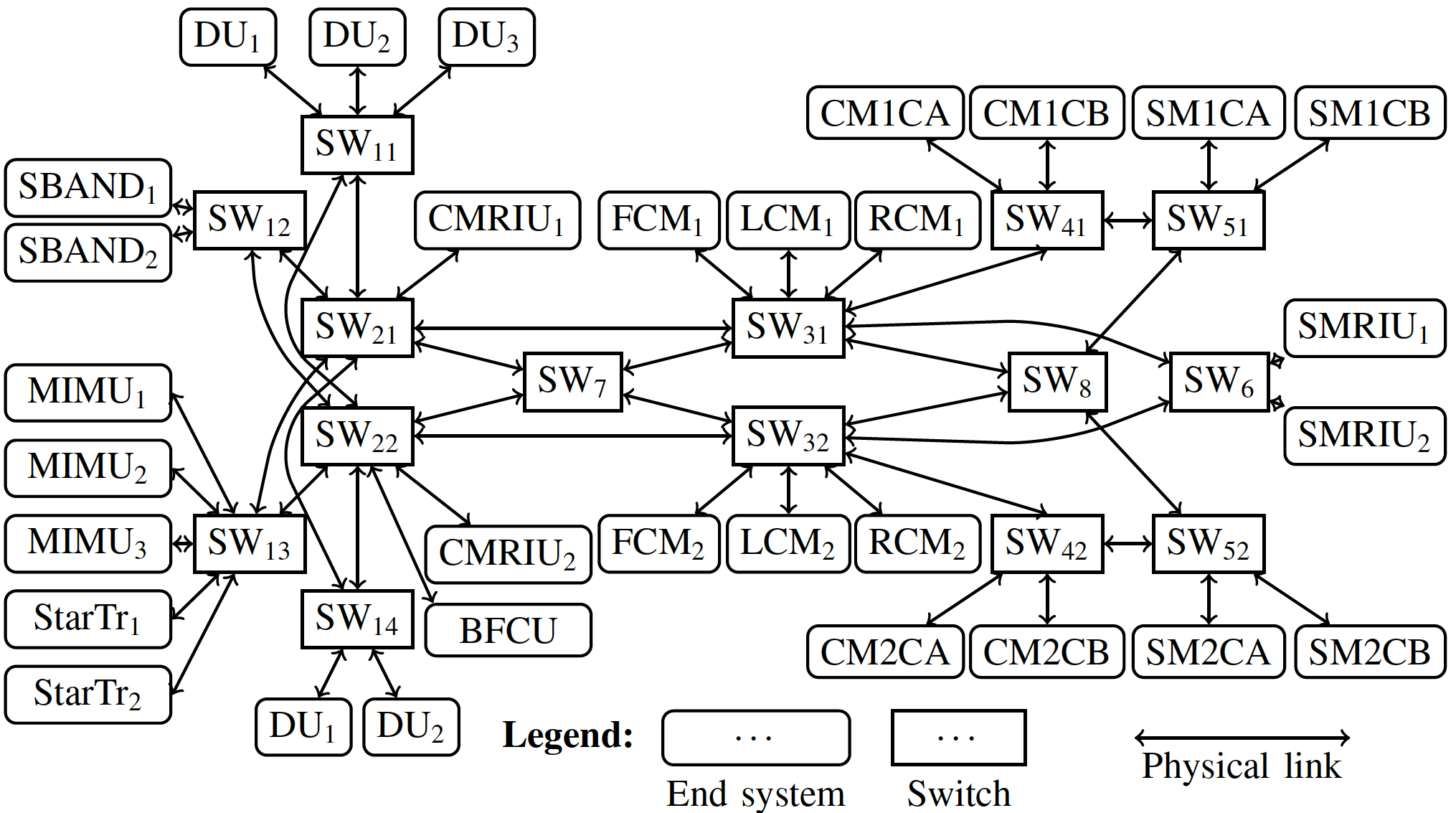}
	\caption{The Orion crew exploration vehicle network, taken from \cite{thomas2020time}.}
	\label{fig:orion}
\end{figure}

\begin{figure*}[h]
	\centering
	
	\includegraphics[width= \textwidth]{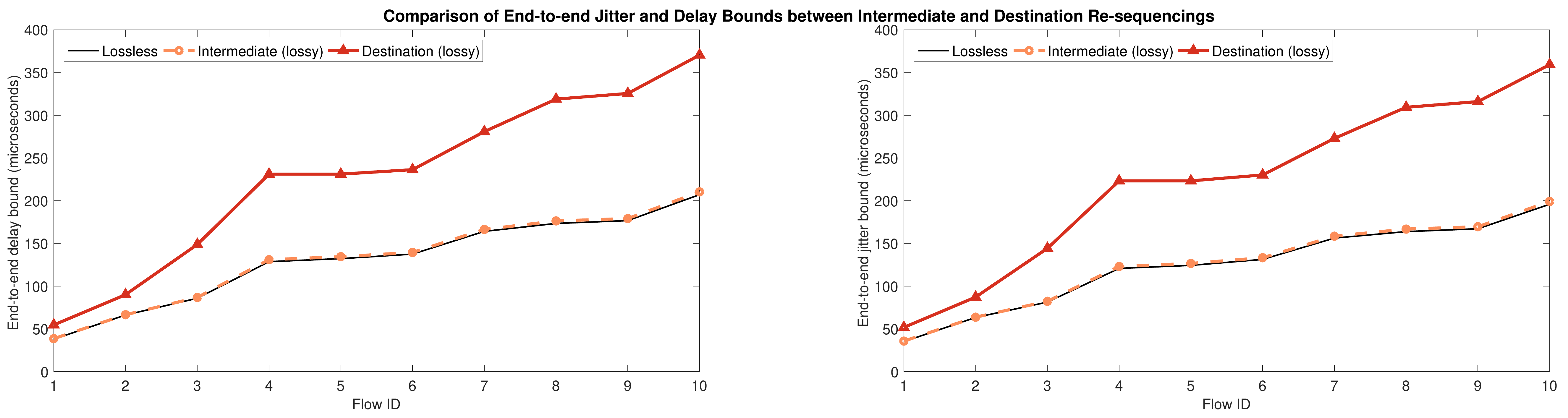}
	\caption{The end-to-end delay bounds (left) and jitter bounds (right) for the flows with in-order delivery requirement for the two strategies, i.e., placing re-sequencing buffers at the destinations or at every switch. In the lossless condition, the delay and jitter bounds do not depend on the strategy.}
	\label{fig:delay_jitter_orion}
\end{figure*} 

Class A contains $30$ flows with constant packet size $147$~bytes, which transmit $3$ packets every $8~$ms. Among these flows, $10$ require in-order packet delivery. The flows traverse between $2$ to $7$ hops. We apply two placement strategies for re-sequencing buffers: at destinations only, and at every switch (immediately after the switching fabrics).

Fig. \ref{fig:delay_jitter_orion} shows the end-to-end delay and jitter bounds for the flows with in-order delivery requirement, for both strategies under lossy condition. The figure also shows the delay and jitter under lossless conditions, which are the same for both strategies and, as we know from Theorem \ref{thm:delay-reseq}, are also equal to the values when there is no re-sequencing buffer. First, we see that if re-sequencing is at destinations only, the effect on delay and jitter under lossy conditions is large: for more than half of the flows, the re-sequencing buffer doubles the delay and jitter, the increase being of the order of 100$\mu$s. This occurs even though the amount of reordering late time offset at every switching fabric is minuscule: every flow has at most $7$ hops and the switching fabric re-orders packet by at most $324$~$n$s at every hop (by Theorem~\ref{thm:gen-node-reordering-time} this is less than the jitter of the switching fabric). This illustrates the pattern of RTO amplification by downstream jitter.  Second, we see that if re-sequencing is performed at every switch, the increase in delay and jitter under lossy conditions is negligible, as expected Section~\ref{sec:irm}, because such a strategy prevents amplification of RTO.

We also find that the size of re-sequencing buffers are the same for the two strategies; this implies that intermediate re-sequencing does not provide any benefit in terms of buffer size. It is equal to two packet size, i.e., $294$~bytes, under lossless network condition; and it is equal to three packet size, i.e., $441$~bytes, under lossy network conditions.
	\section{Conclusion}\label{sec:conclusion}
We have developed a theory of packet reordering in the context of time-sensitive networks, i.e. in networks where worst-cases are more relevant than averages. We showed that, if the network can safely be assumed lossless, re-sequencing does not modify worst-case delay nor delay jitter. In contrast, if performance under lossy operation is relevant, then re-sequencing comes with a penalty on delay equal at least to the RTO of the flow being re-sequenced. We showed that the RTO may be very large even though the RTO of every individual non order-preserving element is very small, due to amplification by downstream jitter. We provided a calculus to capture the RTO and RBO of a flow, given its arrival curve and simple properties of the network elements that are on its path. We applied the theory to evaluate the performance of re-sequencing strategies in industrial networks without re-shaping. Future work will focus on the interactions between re-sequencing and flow-reshaping or flow-damping.

\section{Acknowledgments}
This work was supported by Huawei Technologies Co., Ltd. in the framework of the project Large Scale Deterministic Network. The authors thank Bingyang Liu and Shoushou Ren for fruitful discussions.
	
	\bibliographystyle{IEEEtran}
	\bibliography{ref}
	
	\clearpage
	\appendices
\twocolumn[
\begin{@twocolumnfalse}
 \begin{center}   
{\Large\textbf{Supplementary Material}}\\
{\textbf{On Packet Reordering in Time-Sensitive Networks}\\
\textit{Ehsan Mohammadpour, Jean-Yves Le Boudec}}
\end{center}
  \end{@twocolumnfalse}
  ]
 
 \section{Complementary background}
\subsection{Alternative Characterization of Arrival Curves}\label{sec:appendix-arr}
We will use the following alternative representation of arrival curve constraint \cite{le_boudec_theory_2018}. Consider a flow with packets $1,2,...n, ...$ of sizes $l_1, l_2, ..., l_n, ...$ and let $A_n$ be the observation time of packet~$n$. Assume that the indices are in order of observation, i.e. $A_n$ is wide-sense increasing. Then saying that the flow satisfies the arrival curve constraint $\alpha$ is equivalent to saying that for all indices $m\leq n$:
\begin{align}\label{eq:arrcur-bit}
A_n-A_m\geq \alpha^{\downarrow}\left(\sum_{k=m}^{n}l_k\right)
\end{align} where $\alpha^{\downarrow}$ is the lower-pseudo inverse, defined in Section~\ref{sec:lps}. Similarly, the flow satisfies the packet-level arrival curve constraint $\alpha_{\mathrm{pkt}}$ if and only if for all indices $m\leq n$:
\begin{align}\label{eq:arrcur-pkt}
A_n-A_m\geq \alpha_{\mathrm{pkt}}^{\downarrow}(n-m+1)
\end{align}
  
\subsection{Re-sequencing Buffer Operation}\label{sec:appendix-reseq}
A re-sequencing buffer stores the packets of a flow until the packets with smaller sequence numbers arrive; then it delivers them in the increasing order of their sequence numbers. A re-sequencing buffer has two parameters, a size in bytes, $B$, and a timeout value, $T$.
It is described in terms of:
\begin{algorithm}
	\caption{Packet arrival event code routine}
	\label{alg:receive-event}
	\textbf{Input: }{packet~$p$}\\
	\textbf{Shared variables:}{ $buf$ and $N$}
	\begin{algorithmic}[1]
		\If{$p.id \geq N$} \Comment{if TRUE, $p$ is a valid packet}\label{algLine:valid-packet}
		\If{$p.id > N$}\label{algLine:valid-large}
		\If{$buf.\mathrm{len()} + p.len \leq B$}  \label{algLine:valid-packet-bufsize}
		\State TimerList.start($p.id$,$\mathrm{Time()}+T$) \label{algLine:valid-packet-timerstart}
		\State $buf$.enqueue($p$)\label{algLine:valid-packet-enqueue}
		\Else ~~~ discard($p$)  \Comment{ERROR, OVERFLOW}\label{algLine:valid-packet-overflow}
		\EndIf
		\Else \label{algLine:valid-equal}
		\State $N \gets p.id +1$\label{algLine:valid-equal-N}
		\State release($p$)\label{algLine:valid-equal-release}
		\State CHECK\_BUFFER()
		\EndIf
		\Else ~~ discard($p$) \Comment{ERROR, INVALID PACKET} \label{algLine:invalid-packet}
		\EndIf
		
		\Function{check\_buffer}{void} 	\label{algLine:checkbuffer}
		\If{$buf.\mathrm{contains}(N)$} \label{algLine:checkbuffer-contain}
		\State $p \gets buf$.dequeue($N$) \label{algLine:checkbuffer-dequeue}
		\State $N \gets p.id +1$ \label{algLine:checkbuffer-N}
		\State TimerList.stop($p.id$) \label{algLine:checkbuffer-timerstop}
		\State release($p$) \label{algLine:checkbuffer-release}
		\State CHECK\_BUFFER() \label{algLine:checkbuffer-recursive}
		\EndIf
		\EndFunction
	\end{algorithmic}
\end{algorithm}
\begin{algorithm}
	\caption{Timeout event code routine}
	\label{alg:timeout-event}
	\textbf{Input: }{packet id $pid$}\\
	\textbf{Shared variables:}{ $buf$ and $N$}
	\begin{algorithmic}[1]
		\While{$N \leq pid$}\label{algLine:timeout-mainloop}
		\While{$!buf.\mathrm{contains}(N)$}\label{algLine:timeout-innerloop-start}
		\State $N \gets N +1$
		\EndWhile\label{algLine:timeout-innerloop-end}
		\State $p' \gets buf$.dequeue($N$)\label{algLine:timeout-dequeue}
		\State $N \gets p'.id +1$
		\State TimerList.stop($p'.id$)
		\State release($p'$)\label{algLine:timeout-release}
		\EndWhile
		\State CHECK\_BUFFER()\label{algLine:timeout-checkbuffer}
	\end{algorithmic}
\end{algorithm}

\begin{itemize}
	\item Shared variables, that are manipulated by the code routines. These are 1) a list (buffer), $\mathit{buf}$, containing the packets that are waiting for the packets with smaller sequence number; 2) an integer, $N$, expressing the next sequence number that the buffer is expecting to receive
	\item Timers: The re-sequencing buffer sets a timer for each packet stored in the buffer. The object TimerList is the list of the timers for the packets. It has two functions, start($pid$,$deadline$) and stop($pid$). The former, starts a timer for the packet with sequence number $pid$ with expiration time $deadline$. The latter, stops the timer for the packet with sequence number $pid$.
	\item Events, which trigger the execution of code routines.
	The events are packet arrival (Algorithm~\ref{alg:receive-event}) and timeout (Algorithm~\ref{alg:timeout-event}). We assume that the execution of the code routines is serialized, namely, a code routine can start only after the code routine triggered by the previous event has completed (to avoid race conditions with shared variables).
\end{itemize}

%
%
%
%
%
%
%
%
%
%
%
%
%
%
%
%
%
%
When a new packet~$p$ arrives, the packet arrival code routine in Algorithm \ref{alg:receive-event} is executed.
If the packet sequence number $p.id$ is smaller than $N$, then the packet is considered invalid and is discarded (line \ref{algLine:invalid-packet}). This is an error-case: packet~$N$ is expected, which means that packet~$N-1$ was delivered. This packet is either a duplicate of packet~$N-1$ or a packet with a smaller sequence number, and delivering it would violate in-order delivery.

If $p.id > N$ (line \ref{algLine:valid-large}), the packet should wait in the buffer for the packets with smaller sequence numbers to arrive. Then, it checks the current length of buffer, $buf$.len(); if addition of packet~$p$ with length of $p.len$ does not exceed the size of the buffer, $B$, (line \ref{algLine:valid-packet-bufsize}) a timer for this packet starts, which expires at time Time()+$T$ (line \ref{algLine:valid-packet-timerstart}); the function Time() returns the current time of the buffer. Then, it enqueues the packet in the buffer (line \ref{algLine:valid-packet-enqueue}). Otherwise, if the buffer does not have enough capacity, buffer overflow occurs and the packet is discarded (line \ref{algLine:valid-packet-overflow}).

If $p.id == N$ (line \ref{algLine:valid-equal}), the packet is the expected one. Then, $N$ is incremented by $1$ (line \ref{algLine:valid-equal-N}) and the packet is released (line \ref{algLine:valid-equal-release}). When packet~$p$ departs, the buffer should be checked to release the packets that were only waiting for packet~$p$; this is done by a recursive function CHECK\_BUFFER() (line \ref{algLine:checkbuffer}). Accordingly, if the buffer contains a packet with a sequence number equal to the new value of $N$ (line \ref{algLine:checkbuffer-contain}), the packet is dequeued from the buffer (line \ref{algLine:checkbuffer-dequeue}). Then, $N$ is increased to the next sequence number of the flow; the corresponding timer is stopped; and then the packet is released (lines \ref{algLine:checkbuffer-N} to \ref{algLine:checkbuffer-release}). It recalls the functions recursively (line \ref{algLine:checkbuffer-recursive}) and the value of $N$ is increased every time the function executed; the recursion continues until the buffer does not have a packet with sequence number equal to the last updated value of $N$; all the packets with sequence numbers less than the value of $N$ are already released from the buffer.

When the timer for a packet with sequence number $pid$ expires, the timeout code routine in Algorithm \ref{alg:timeout-event} is executed. In this condition, the packet with sequence number $pid$ should be released. To provide in-order delivery, the buffer first releases all the packets in the buffer with sequence number less than or equal to $pid$ in increasing order. To do so, the loop at line \ref{algLine:timeout-mainloop} is executed to iterate the buffer for the packets with sequence number less than $pid$.
The lines \ref{algLine:timeout-innerloop-start} to \ref{algLine:timeout-innerloop-end}, increment $N$ to the smallest sequence number for a packet stored in the buffer; the function $\mathit{buf}$.contains($N$), returns TRUE if a packet with sequence number $N$ is in the buffer, otherwise, it returns FALSE. In line \ref{algLine:timeout-dequeue}, the buffer contains a packet~$p'$ with sequence number equal to $N$; then, it is dequeued from the buffer; the value of $N$ is incremented by $1$; the timer for this packet is stopped; and it is released from the buffer (lines \ref{algLine:timeout-dequeue} to \ref{algLine:timeout-release}). Whenever a packet is released from the buffer, its corresponding timer is stopped. The loop in line \ref{algLine:timeout-mainloop} is executed at least once, because when the timeout event occurs for a packet with sequence number $pid$, it is already in the buffer; this implies $N\leq p.id$. The loop in line \ref{algLine:timeout-innerloop-start} is executed at the latest when $N =pid$.
In line \ref{algLine:timeout-checkbuffer}, the buffer is checked to release the packets that were waiting for packets with id $pid$ and smaller; this is done by the function CHECK\_BUFFER().

Observe that a packet is released if any one of the following conditions hold: 1) all packets with smaller sequence numbers are received, 2) its timer expires or 3) the timer of a received packet with a larger sequence number expires.

By construction, the re-sequencing buffer delivers the packets that it does not discard in increasing sequence numbers. Furthermore, a packet is discarded by the re-sequencing buffer either when the buffer is full or when the sequence number of the arriving packet is less than $N$. The latter occurs when the timeouts of packets are too early, compared to the lateness of misordered packets. Therefore, to avoid discarding packets, the re-sequencing buffer size and the timeout value should be large enough.

\section{Proofs}

\subsection{Proof of Lemma~\ref{lem:acp}}

Let $A_1 \leq A_2\leq ...\leq A_n...$ be the arrival times of packets at $\calS$; here the packet indices are in order of arrival (and are not necessarily equal to the sequence numbers). Let $D_n$ be the departure time of the packet with index $n$; the sequence $D_n$ need not be monotonic. Let $R,R'$ be the cumulative arrival and departure functions\cite{le_boudec_theory_2018}, defined by
$R(t)=\sum_{n=1}^{+\infty}l_n \ind{A_n<t}$ and $R'(t)=\sum_{n=1}^{+\infty}l_n \ind{D_n<t}$. Here, $l_n$ is the length of packet~$n$ and we allow $t\leq 0$ (in which case $R(t)=R'(t)=0$).
Let $\varphi(t)$ be the function $\Reals\to\Reals$ defined by $\varphi(t)=0$ is $t\leq 0$ and $\varphi(t)=1$ if $t>0$, so that
\begin{equation}\label{eq:l11}
  R(t)=\sum_{n=1}^{+\infty}l_n \varphi(t-A_n), \;\;\; R'(t)=\sum_{n=1}^{+\infty}l_n \varphi(t-D_n)
\end{equation}
Let $0\leq s\leq t$; the arrival curve constraint at the input means that $R(t)-R(s)\leq \alpha(t-s)$.
This equation continues to hold if $s$ or $t$ is negative, with the convention that $\alpha(t)=0$ whenever $t\leq0$. Therefore
\begin{equation}\label{eq:l12}
  \forall s,t \in \Reals,\;\;\;R(t)-R(s)\leq \alpha(t-s)
\end{equation}
Furthermore,
\begin{equation}\label{eq:l13}
  R'(t)-R'(s)= \sum_{n=1}^{+\infty}l_n\left(\varphi(t-D_n)-\varphi(s-D_n)\right)
\end{equation}

Let $d^{\min}$ be the best-case delay of the flow, so that the worst-case delay is $\leq d^{\min}+V$. For every packet index $n$, we have
\begin{align}\label{eq:l133}
  A_n+d^{\min} \leq D_n  \leq A_n + d^{\min}+V
\end{align}
thus
\begin{align}\label{eq:l144}
  t- D_n  &\leq t - A_n - d^{\min}\\
  s-D_n &\geq s-A_n - d^{\min}-V
\end{align} and, since $\varphi$ is wide-sense increasing
\begin{align}\label{eq:l155}
  \nonumber\varphi(t- D_n)-\varphi(s-D_n )   \leq
  &\varphi(t - A_n - d^{\min})\\
  &- \varphi(s-A_n - d^{\min}-V)
\end{align}
Combining with \eqref{eq:l13} and \eqref{eq:l12}:
\begin{align}\label{eq:l66}
  \nonumber R'(t)-R'(s)&\leq \sum_{n=1}^{+\infty}l_n\Big(\varphi(t - A_n - d^{\min})\\
  &\nonumber-  \varphi(s-A_n - d^{\min}-V)\Big)\\
  &\nonumber =  R(t - d^{\min})-R(s - d^{\min}-V)\\
  &\leq \alpha(t-s+V)
\end{align}
\qed

\subsection{Proof of Lemma~\ref{lem:backlog}}

With the same notation as in the proof of Lemma~\ref{lem:acp}, the backlog of this flow at time $t$ is $B(t)=R(t)-R'(t)$, so that
\begin{equation}
  B(t)=\sum_{n=1}^{+\infty}l_n\left(\varphi(t-A_n)-\varphi(t-D_n)\right)
\end{equation}
We have $  t-D_n \geq t-A_n - U$ and, since $\varphi$ is wide-sense increasing:
\begin{align}\label{eq:l233}
  \nonumber B(t) & \leq \sum_{n=1}^{+\infty}l_n\left(\varphi(t-A_n)-\varphi(t-A_n-U)\right) \\
    & = R(t)-R(t-U)\leq \alpha(U)
\end{align}
where the last inequality is by \eqref{eq:l12}.
\qed

\subsection{Proof of Theorem~\ref{thm:io}}  \label{sec:appendix_thm_io}
	(1) We prove \eqref{eq:reseq-io} by induction on $n$.
	
	Base case $n=1$. According to the description of re-sequencing buffer in Section \ref{sec:sys}, the initial value of next expected sequence number is $N=1$; therefore, packet~$1$ is never stored in the buffer.

	\begin{itemize}
		\item If packet~$1$ arrives at some time $E_1 \neq +\infty$. Then, there are two cases for the timers at time $E_1$ of packets with index larger than 1:
		\begin{itemize}
			\item No timer has expired at $E_1$. Then, $N=1$ and packet~$1$ is released immediately after arrival ($D_1 = E_1$). Since no timer was expired, for any packet~$p$ in the buffer, $p>1$, we have $E_1 \leq E_p + T$. Thus, based on \eqref{eq:ioi}, $I_1 = E_1$ and finally $D_1=I_1$ as required.
			\item A timer has expired, say for packet index $p$, at $E_1$. Then, $N>1$ and packet~$1$ is discarded ($D_1 = +\infty$). We have $E_1 > E_p + T$, \eqref{eq:ioi} gives $I_1 = +\infty$ and finally $D_1=I_1$ as required.
		\end{itemize}
		\item Else packet~$1$ is lost ($E_1 = +\infty$). Then it is not released from the buffer ($D_1 = +\infty$). By \eqref{eq:ioi}, $I_1 =E_1 = +\infty$ and finally $D_1=I_1$ as required.
	\end{itemize}

	Induction step. Suppose that \eqref{eq:reseq-io} holds for all $i, i\leq n-1$. 	
\begin{itemize}
		\item If packet~$n$ arrives at some time $E_n\neq +\infty$.
		Then, there are two cases for the value of $N$:
		\begin{itemize}
			\item $N\leq n$. No timer for a packet~$j$ in the buffer, $j>n$, has been expired (otherwise, the value of $N$ would be increased to the index of the packet with expired timer, i.e., $N>n$). Then, we have $E_n \leq E_j + T$, therefore, $E_n \leq T + \min_{j\geq n}\{E_j\}$. Then, based on \eqref{eq:ioi}, $I_n = E_n$. Now, the release time of packet~$n$ depends on the status of packet~$n-1$; there are two possible cases:
			\begin{itemize}
				\item $N=n$. This implies that packet~$n-1$ is already released, $D_{n-1} \leq E_n$. Therefore, packet~$n$ is released immediately on arrival ($D_n = E_n$). Then, based on \eqref{eq:reseq-io} and \eqref{eq:ioi}, we also have:
				\begin{align}
				I_n &= E_n,\\
				 \nonumber G_n &=\min\left\{D_{n-1}, T + \min_{j\geq n}\{E_j\}\right\} \\
				 &\leq D_{n-1} \leq E_n,
				\end{align} and then $\max(G_n , I_n) = E_n$, which shows that the right hand-side of \eqref{eq:reseq-io} is $D_n$ as required.
				\item $N<n$. This implies that packet~$n-1$ is not yet released, $D_{n-1} > E_n$. Therefore, packet~$n$ is stored in the buffer. Packet~$n$ is released when:
				\begin{itemize}
					\item The timer for a packet with sequence number larger than or equal to $n$ is expired before packet~$n-1$ is released ($D_{n-1} \geq \min_{j\geq n}\{E_j\} + T$). Let us call $p$, where $p\geq n$, as the packet of which the timer is expired before the others ($T + \min_{j\geq n}\{E_j\} = T+ E_p$). Then, any packet~$j$ in the buffer $j\leq p$ are released in-order (Algorithm \ref{alg:timeout-event}, line \ref{algLine:timeout-release}); therefore, $D_n = T+ E_p = T + \min_{j\geq n}\{E_j\}$. Then, based on \eqref{eq:reseq-io} and \eqref{eq:ioi}, we also have:
					\begin{align}
					I_n &= E_n,\\
					\nonumber G_n &=\min\left\{D_{n-1}, T + \min_{j\geq n}\{E_j\}\right\} \\
					 &= T + \min_{j\geq n}\{E_j\},
					 \end{align}
					 and then we have:
					 \begin{align}
					\nonumber \max\{G_n , I_n\} &= \max\{T + \min_{j\geq n}\{E_j\} , E_n\} \\
					&= T + \min_{j\geq n}\{E_j\}=D_n,
					\end{align}
which shows \eqref{eq:reseq-io}.
					\item Or when packet~$n-1$ is released before any timer of packets with sequence number larger than or equal to $n$ is expired ($D_{n-1} < T+\min_{j\geq n}\{E_j\}$). Then packet~$n$ is released immediately after packet~$n-1$ is released ($D_n = D_{n-1}$). Then, based on \eqref{eq:reseq-io} and \eqref{eq:ioi}, we also have:
					\begin{align}
					G_n &= \min\left\{D_{n-1}, T + \min_{j\geq n}\{E_j\}\right\} \\
					&= D_{n-1}.
					\end{align}
					Then,
					\begin{align}
					\nonumber \max\{G_n , I_n\} &= \max\{G_n , E_n\} \\
					&= D_{n-1} =D_n,
					\end{align}

which shows \eqref{eq:reseq-io}.
				\end{itemize}
			\end{itemize}
		\item $N>n$. Then, packet~$n$ is discarded ($D_n = +\infty$). Since $N>n$, a timer should has been expired for a packet~$p$ in the buffer, $p>n$ such that $E_n > E_p + T$; Then, based on \eqref{eq:reseq-io} and \eqref{eq:ioi}, we also have:
			\begin{align}
			I_n = +\infty, \max\{G_n, I_n\} = +\infty=D_n,
			\end{align}which shows \eqref{eq:reseq-io}.
		\end{itemize}
		\item Else, packet~$n$ is lost ($E_n = +\infty$). Then it is not released from the buffer ($D_n = +\infty$). By \eqref{eq:ioi}, $I_n = +\infty$ as well and thus $\max\{G_n, I_n\} = +\infty$, i.e. the right-handside of \eqref{eq:reseq-io} is equal to $D_n$ as required.
	\end{itemize}

(2) 	Consider Fig. \ref{fig:reseq-gen} and a received packet~$n$. Due to \eqref{eq:reordering-time-def}, for any packet~$j\geq n$, we have $E_n \leq E_j + \lambda$; therefore, $E_n\leq \min_{j\geq n}\{E_j\}+\lambda$. Since $\lambda \leq T$:
	\begin{align}
	E_n \leq \min_{j\geq n}\{E_j\}+T \implies I_n = E_n,
	\end{align}
	which proves \eqref{eq:iod2}.


(3) The proof of item (3) is by induction on $n\geq 1$.
	
	Base case $n=1$. Then, by \eqref{eq:iod2}, $D_1 = E_1$ and the statement is trivially proven.
	
	Induction step. We assume that the statement holds for all packets $i$ with $i\leq n-1$. Due to \eqref{eq:reordering-time-def}:
	\begin{align}
	\nonumber &&\forall k<n, \forall p\geq n:~~&E_{k} \leq E_p + \lambda\\
	\nonumber & \mbox{ thus } &\forall k<n:~~ &E_{k} \leq T + \min_{j\geq n}\{E_j\},\\
	\nonumber &\mbox{ thus } &&\max_{k\leq n-1}\left\{E_k\right\} \leq T + \min_{j\geq n}\{E_j\} \\
	&\mbox{ thus } &&D_{n-1} \leq T + \min_{j\geq n}\{E_j\}.
	\end{align}
	By \eqref{eq:iog} and \eqref{eq:iod2}:
	\begin{align}
	 G_n &=\min\left\{D_{n-1}, T + \min_{j\geq n}\{E_j\}\right\} = D_{n-1},\\
	\nonumber D_n &= \max(G_n , E_n) =\max(D_{n-1} , E_n) \\
	&= \max\left(\max_{k\leq n-1}\left\{E_k\right\} , E_n\right) = \max_{k\leq n}\left\{E_k\right\}.
	\end{align}

\subsection{Proof of Theorem~\ref{thm:reseq-param}} \label{sec:appendix_thm_reseq-timeout}

	Consider Fig. \ref{fig:reseq-gen}.
	
	First, we prove by induction on $n\geq 1$ that
\begin{equation}D_n \geq \max_{i<n|E_i\neq+\infty}\{D_i\} \label{eq-asfrdgh}\end{equation}
	
	Base case. By Theorem~\ref{thm:io}, $n=1$, $D_1 = E_1$, and then \eqref{eq-asfrdgh} is obvious.
	
	Induction step. We assume that \eqref{eq-asfrdgh} holds for all packet index $i<n$. According to Theorem~\ref{thm:io}, for packet~$n\geq 2$, we have:
	\begin{align}\label{eq:reseq-param-1}
	D_n &= \max(G_n , E_n),\\
	G_n &= \min\left(D_{n-1}, T + \min_{j\geq n}\{E_j\}\right).
	\end{align}
	Consider the set $\calE_n=\left\{i<n| E_i\neq +\infty\right\}$, i.e. the packet numbers less than $n$ that are not lost in the network. If $\calE_n$ is empty, the set is empty, \eqref{eq-asfrdgh} trivially holds. Therefore, we now assume that $\calE_n$ is not empty. Let $m$ be the maximum of $\calE_n$.
By the induction hypothesis, $D_m = \max_{k<m|E_k\neq+\infty}\{D_k\}$. Then \eqref{eq-asfrdgh} gives:
\begin{align}
\nonumber D_n &\geq \max_{k<n|E_k\neq+\infty}\{D_k\} \\
\nonumber &= \max\left(\max_{k<m|E_k\neq+\infty}\{D_k\}, \max_{m\leq k<n|E_k\neq+\infty}\{D_k\}\right)\\
&= \max\left(D_m,\max_{m\leq k<n|E_k\neq+\infty}\{D_k\}\right).
\end{align}
Since $m = \max\{\calE_n\}$, we have $\max_{m\leq k<n|E_k\neq+\infty}\{D_k\} = D_m$. Therefore we need to show $D_n \geq D_m$ to prove the theorem.
Due to RTO bound for packet~$m$:
	\begin{align}
	\nonumber &&\forall j> m:~~&E_{m} \leq E_j + \lambda \leq E_j + T,\\
	&\mbox{ thus } &&E_{m} \leq T + \min_{j> m}\{E_j\} .
	\end{align}
	Since $E_{m} \leq T + E_{m}$, then:
	\begin{align}\label{eq:reseq-param-5}
	E_{m} \leq T + \min_{j\geq m}\{E_j\}.
	\end{align}
	By part (2) of Theorem \ref{thm:io} for packet~$m$, we have:
	\begin{align}\label{eq:reseq-param-2}
	\nonumber G_{m} &= \min\left(D_{m-1}, T + \min_{j\geq m}\{E_j\}\right) \leq T + \min_{j\geq m}\{E_j\},\\
	\nonumber D_{m} &= \max(G_{m} , E_{m}) \leq \max\{T + \min_{j\geq m}\{E_j\} , E_{m}\} \\
	&\stackrel{\text{Eq. \eqref{eq:reseq-param-5}}}\implies D_{m}\leq T + \min_{j\geq m}\{E_j\}.
	\end{align}
	We consider the two possible cases for $m$: 1) $m = n-1$, 2) $m<n-1$.
	\begin{itemize}
	\item $m=n-1$. Since $n-1<n$, $\min_{j\geq n-1}\{E_j\} \leq \min_{j\geq n}\{E_j\}$. Therefore, by \eqref{eq:reseq-param-2}, $D_{n-1} \leq T+\min_{j\geq n}\{E_j\}$. Now, using \eqref{eq:reseq-param-1}:
		\begin{align}\label{eq:reseq-param-4}
		\nonumber D_n = \max(G_n , E_n) \geq G_n &= \min\left(D_{n-1}, T + \min_{j\geq n}\{E_j\}\right) \\
		&= D_{n-1}.
		\end{align}
	\item $m<n-1$. Then $E_{n-1} = +\infty$ and in turn, $D_{n-1} = \max\{G_{n-1} , E_{n-1}\} = +\infty$. By \eqref{eq:reseq-param-1}, we have:
		\begin{align}\label{eq:reseq-param-3}
		\nonumber D_n = \max(G_n , E_n) \geq G_n &= \min\left(D_{n-1}, T + \min_{j\geq n}\{E_j\}\right) \\
		&= T + \min_{j\geq n}\{E_j\}.
		\end{align}
		Since $m<n-1$, $\min_{j\geq m}\{E_j\} \leq \min_{j\geq n}\{E_j\}$. Therefore using \eqref{eq:reseq-param-2} and \eqref{eq:reseq-param-3}, we have:
		\begin{align}
		D_n \geq T + \min_{j\geq n}\{E_j\} \geq T + \min_{j\geq m}\{E_j\} \geq D_{m}.
		\end{align}
	\end{itemize}
This establishes \eqref{eq-asfrdgh}.

Second, we show that if for any $\lambda >0$, if $T<\lambda$, there exists a scenario with RTO $\lambda$
where the re-sequencing buffer discards a packet
Consider a trace with two packets $1$ and $2$, received at the re-sequencing buffer at
a times $E_2 = t_0$ and $E_{1} = t_0 + \lambda$ for some $t_0\geq 0$. The RTO of this trace is $\lambda$. By Theorem~\ref{thm:io}, $D_1=I_1$ and $I_1=+\infty$ because $E_1>\min_{j\geq n}\{E_j\}+T=E_2+T$. Thus packet $1$ is discarded by the re-sequencing buffer.


\subsection{Proof of Theorem~\ref{thm:reseq-param-size}} \label{sec:appendix_thm_reseq-size}
\textbf{Item (1)}. Consider Fig. \ref{fig:reseq-gen}.
Assume the size of the re-sequencing buffer is unlimited; the actual buffer content at time $t$ is
\begin{align}\label{eq:def-L}
	L(t) = \sum_{E_k<t,D_k\geq t}l_k,
	\end{align}
First, we show that $L(t)\leq \pi$ at all times $t$ that immediately follow a packet arrival; this will imply that a buffer of size $\pi$ is sufficient to avoid overflow.
	
Packet $1$ is never stored in the buffer. Consider some fixed but arbitrary packet $n>1$, with size $l_n$, and define the set of indices $\calX$ by
\begin{equation}\calX\eqdef  \left\{i\in\mathbb{Z}^+~|~i<n,  D_i\geq E_n\right\}\end{equation}
If $\calX$ is empty then $D_{n-1}<E_n$; observe that $D_n=\max (E_1, ..., E_n)=\max(D_{n-1}, E_n)=E_n$, i.e. packet $n$ is not stored in the buffer, and therefore buffer overflow does not occur when packet $n$ arrives. Hence, we assume that $\calX$ is not empty and let $m=min (\calX)$. The actual content of the buffer just after the arrival of packet $n$ is
\begin{align}
  L(&E_n)+l_n  = \sum_{k, E_k<E_n\leq D_k}  l_k + l_n \label{eq-pr33fdsj}\\
    &= \sum_{k<m, E_k<E_n\leq D_k} l_k  + \sum_{k\geq m, E_k<E_n\leq D_k} l_k   + l_n \label{eq-pr3342}\\
      &=\sum_{k\geq m, E_k<E_n\leq D_k} l_k + l_n \leq \sum_{k\geq m, E_k<E_n} l_k + l_n, \label{eq-pr3343}
\end{align}
where the last equality is because the first sum in \eqref{eq-pr3342} is $0$ by definition of $m$. By Lemma \ref{lem:reseq-buffer}, $E_m \geq E_n$; since $m\neq n$ and we exclude simultaneous packet arrivals at the re-sequencing buffer, $E_m > E_n$. Thus
\begin{align}\label{eq-pr33klsfd}
 \nonumber  \Big[(k\geq m \mbox{ and } &E_k<E_n) \mbox{ or } (k=n) \Big]\\
  & \implies (k> m \mbox{ and } E_k<E_m),
\end{align}
therefore
\begin{align}\label{eq-pr33-sdjfh}
  \nonumber\sum_{k\geq m, E_k<E_n} l_k + l_n &= \sum_{k, (k\geq m, E_k<E_n) \mathrm{ or } (k=n)}\\
  &\leq  \sum_{k> m, E_k<E_m} l_k = \pi_m.
\end{align}
Combined with \eqref{eq-pr33fdsj}-\eqref{eq-pr3343}, this shows that
\begin{equation}\label{eq-pr33-22}
  L(E_n)+l_n  \leq \pi_m\leq \pi
\end{equation}

\begin{lemma}\label{lem:reseq-buffer}
	$E_m \geq E_n$.
\end{lemma}
	\begin{proof}
By construction, $m\in \calX$ thus $D_m\geq E_n$. Since $D_m=\max_{k\leq m}E_k$, it follows that
\begin{equation}\label{eq-l3-1}
  \exists k \in \left\{1...m\right\} \mbox{ such that }E_k\geq E_n
\end{equation}
If $m=1$ the conclusion follows.
Else, $m-1$ is not in $\calX$ thus $D_{m-1}<E_n$.
Since $D_{m-1}=\max_{k\leq m-1}E_k$, it follows that $\forall k \in \left\{1...m-1\right\}$, $E_k <E_n$. Combined with \eqref{eq-l3-1}, this shows that $E_m\geq E_n$.
	\end{proof}

Second, we show that for any possible $\lambda>0$ and valid RBO value $\pi$ there exists one execution trace of a flow with packet sizes between $L^{\min}$  and $L^{\max}$, with RTO $\lambda$ and RBO $\pi$, that achieves a buffer content equal to $\pi$. First observe that, by definition, $\pi$ can be written as $\pi=\sum_{j=1}^k \ell_j$ for some positive integer $k$ and  $\ell_j \in [L^{\min}, L^{\min}]$. The packet sequence is as follows. It has $k+1$ packets in total. Packet $1$ has some arbitrary size $l_1 \in [L^{\min}, L^{\min}]$ and is observed at time $E_1=\lambda$. Packets 2 to $k+1$ have sizes $l_2=\ell_1, ...l_{k+1}=\ell_k$ and are observed at times $E_j=\frac{(j-1)\lambda}{k+1}$. Packets $2$ to $k+1$ arrive before packet $1$, are stored in the buffer until packet $1$ arrives, and the buffer content when packet $1$ arrives is $\sum_{j=1}^k l_j$. The RTOs are $\lambda_1=\lambda, \lambda_2=...=\lambda_{k+1}=0$ and the RTO of the trace is $\lambda$. The RBOs are $\pi_1= \sum_{j=1}^k \ell_j=\pi$, $\pi_2=...=\pi_{k+1}=0$ thus the RBO of the sequence is $\pi$.

	\noindent\textbf{Item (2)}. First we show that the buffer size is upper bounded by $\alpha(T+V)$. Since delay jitter from the source to the input of the re-sequencing buffer is $V$, by Lemma \ref{lem:acp}, the flow has arrival curve $\alpha'(t) = \alpha(t+V)$ at the input of re-sequencing buffer. Also, by Theorem~\ref{thm:delay-reseq} (the proof of which is independent of this result) the delay at the re-sequencing buffer is upper bounded by the time-out $T$; by Lemma \ref{lem:backlog}, the amount of backlog inside the buffer is thus upper bounded by $\alpha'(T) = \alpha(V+T)$.

Fix some $\varepsilon>0$, smaller than $V$ and $T$. By the second technical assumption at the end of Section~\ref{sec:arrcur}, there exists an integer $n$ and a sequence of packet lengths $\ell_k\in [L^{\min}, L^{\max}]$ such that $\alpha(V+T-\varepsilon)=\sum_{k=1}^n \ell_k$. Since the arrival curve is achievable, there also exists a sequence of emission times $t_1=0, ...t_n=V+T- \varepsilon$ such that the packet sequence $t_1, ...t_n ,\ell_1, ..., \ell_n$ satisfies the arrival curve constraint $\alpha$. We now derive another packet sequence of $n+1$ packets as follows.
\begin{enumerate}
  \item Packet $1$ is emitted at time $A_1=0$ and has size $l_1=L^{\max}$.
  \item For $k=2...n+1$, packet $k$ is emitted at time $A_k=t_0+t_{k-1}$ and has size $l_k=\ell_{k-1}$, where $t_0$ is a positive number, large enough so that $\alpha(t_0)\geq \sum_{k=1}^{n+1}l_k$. Such a number exists because we assume $\lim_{t\to\infty}\alpha(t)=+\infty$. We have thus $A_{n+1}-A_2=V+T- \varepsilon$.
\end{enumerate}
The arrival times of the $n+1$ packets to the input of re-sequencing buffer are as follows.

\begin{enumerate}
  \item Packet $1$ is lost, i.e. $E_1=+\infty$.
  \item Packet $k=2$ arrives at time $E_2=V+A_2-\frac{\varepsilon}{4}$.
  \item If $n\geq 2$, for $k=3...n+1$, packet $k$ arrives at time $E_k=\max(E_{k-1},A_k)+ \frac{\varepsilon}{3(n-1)}$. 
\end{enumerate}
We now verify that our scenario satisfies all constraints. There is no simultaneous arrival at the re-sequencing buffer as required by our modelling assumptions. Obviously $E_k\geq A_k$ (the scenario is causal) and $E_k>E_{k-1}$ for $k\geq 3$ therefore there is no reordering, and any RTO or RBO constraint is satisfied.

We now verify that the jitter is $\leq V$. We first show by induction on $k\geq 2, k\leq n+1$ that
\begin{equation}\label{eq-pr32199302}
   E_k-A_k\leq V-\frac{\varepsilon (k-2)}{3 (n-1)}
\end{equation} For $k=2$ it follows from the definition of $E_2$. Consider now $k\geq 3$ and assume it holds for $k-1$. Then, by the induction hypothesis:
\begin{align}\label{eq-pr3123a}
 \nonumber &&E_{k-1}&\leq A_{k-1}+V-\frac{\varepsilon (k-3)}{3 (n-1)} \\
 \nonumber & &&\leq A_k+V-\frac{\varepsilon (k-3)}{3 (n-1)} \\
 &\mbox{thus } &E_{k-1}+ \frac{\varepsilon}{3(n-1)}&\leq A_k+V-\frac{\varepsilon (k-2)}{3 (n-1)}.
\end{align}
Also, as $V>\varepsilon$:
\begin{align}\label{eq-pr3123c}
A_{k}+ \frac{\varepsilon}{3(n-1)}&\leq A_k+V-\frac{\varepsilon (k-2)}{3 (n-1)}.
\end{align}
Then, \eqref{eq-pr3123a} and \eqref{eq-pr3123c} give:
\begin{align}
E_k=\max(E_{k-1},A_k)+ \frac{\varepsilon}{3(n-1)}&\leq A_k+V-\frac{\varepsilon (k-2)}{3 (n-1)},
\end{align}
as required. It follows from \eqref{eq-pr32199302} that the jitter of the trace is less than or equal to $V$.

Next, packet $2$ arrives at time $V+A_2-\frac{\varepsilon}{4}$ and is out of order (due to the loss of packet $1$), which triggers a timeout at time $T+V+A_2-\frac{\varepsilon}{4}$.
We now verify that all packets $k\geq 3$ arrive before $V+T+A_2-\frac{\varepsilon}{4}$. To this end, we show by induction on $k\geq 2$ that
\begin{equation}\label{eq-pr32199303}
   E_k\leq T+ V+A_2-\frac{\varepsilon (2n-k)}{3 (n-1)}.
\end{equation} For $k=2$, it follows from the definition of $E_2$ and $T>\varepsilon$. Assume it holds for $k-1$. Then, by the induction hypothesis
\begin{align}\label{eq-pr3123b}
  \nonumber E_{k-1}&\leq T+ V+A_2-\frac{\varepsilon (2n-k+1)}{3 (n-1)} \\
 \mbox{thus } E_{k-1}+ \frac{\varepsilon}{3(n-1)}&\leq T+ V+A_2-\frac{\varepsilon (2n-k)}{3 (n-1)}.
 \end{align}
 Also $A_{k}\leq A_{n+1} = T+ V+A_2-\varepsilon$, therefore,
 \begin{align}
\nonumber  A_k + \frac{\varepsilon}{3(n-1)}&\leq T+ V+A_2-\frac{\varepsilon (3n-2)}{3(n-1)}\\
&\leq T+ V+A_2-\frac{\varepsilon (2n-k)}{3 (n-1)},
\end{align}
hence,
\begin{align}
E_k&\leq T+ V+A_2-\frac{\varepsilon (2n-k)}{3 (n-1)},
\end{align}
as required.

Now $\frac{\varepsilon (2n-k)}{3 (n-1)}>\frac{\varepsilon}{4}$ for $k=3...n+1$ thus every packet other than $2$ arrives before packet $2$ times out. Thus the buffer content just after the arrival of packet $n+1$ is all packets $2$ to $n+1$, i.e. its size is $\alpha(V+T-\varepsilon)=\sum_{k=1}^n \ell_k$.

It remains to verify that the trace satisfies the arrival curve constraint. Let $R(t)$ be the cumulative arrival function of the trace, i.e. $R(t)=\sum_{n=1}^{+\infty}l_n \ind{A_n<t}$.
First, the sequence of packets $2$ to $n+1$ is obtained by time-shifting by $t_0$ a sequence that satisfies the arrival curve constraint, therefore it also does, namely, $R(t)-R(s)\leq \alpha(t-s)$ whenever $s\leq t$, $s\geq t_0$ and $t\geq t_0$. It remains to see the other cases:
\begin{itemize}
  \item $0<s\leq t<t_0$: then $R(t)-R(s)=0\leq \alpha(t-s)$
  \item $0=s\leq t<t_0$: then $R(t)-R(s)=R(t)-R(0)=l_1\leq \alpha(0+)\leq \alpha(t-s)$
  \item $0<s<t_0\leq t$: then $R(t)-R(s)=R(t)-R(t_0)\leq \alpha(t-t_0)\leq \alpha(t-s)$
  \item $0=s<t_0\leq t$: then $R(t)-R(s)=R(t)-R(0)=R(t)-R(t_0)+l_1\leq \sum_{k=1}^{n+1}l_k\leq \alpha(t_0)$ by construction of $t_0$. Thus $R(t)-R(0)\leq \alpha(t_0)\leq \alpha(t)$.
\end{itemize}
This shows that the arrival curve constraint is satisfied.

At this stage, we have shown that, for every $\varepsilon$ small enough, there is a scenario where the backlog reaches $\alpha(V+T-\varepsilon)$. Thus the minimal bound is at least $\sup_{\varepsilon>0}\alpha(V+T-\varepsilon)=\alpha(V+T)$ because arrival curves are left-continuous.

\subsection{Proof of Theorem~\ref{thm:delay-reseq}} \label{sec:appendix_thm_delay_reseq}

\begin{proof}
	We use the notations in Fig. \ref{fig:reseq-gen}. Suppose that the system has $\delta^{\max}$ and $\delta^{\min}$ as worst-case and best-case delays, thus, for any packet~$i$, $\delta^{\min} \leq E_i - A_i \leq \delta^{\max}$. Now, consider a received packet~$n$.	
	Since $D_n \geq E_n$,  $D_n - A_n \geq E_n -A_n \geq \delta^{\min}$; this shows that the best-case delay is not decreased.
	
	Part (1). The system is lossless, therefore by part (3) of Theorem \ref{thm:io}, $D_n = \max_{i\leq n}\{E_i\}$. Therefore:
	\begin{align}
	D_n - A_n = \max_{i\leq n}\{E_i\} - A_n = \max_{i|i\leq n}\{E_i - A_n\}.
	\end{align}
	Since $\forall i\in \mathbb{Z}^+, A_i \leq A_{i+1}$:
	\begin{align}
	D_n - A_n\leq \max_{i\leq n}\{E_i - A_i\} \leq \delta^{\max},
	\end{align}
	which proves that the worst-case delay is not increased. Since the best case delay is not decreased, the delay jitter is not increased.
	
	Part (2). The system is not lossless. By part (2) of Theorem \ref{thm:io}, we have:
	\begin{align}
	\nonumber G_n &= \min\left(D_{n-1}, T+\min_{j\geq n}E_j\right) \leq T+\min_{j\geq n}E_j \leq T + E_n,\\
	D_n &= \max(G_n,E_n) \leq  \max(T + E_n,E_n) = T + E_n.
	\end{align}
	Thus
	\begin{align}
	D_n -A_n \leq T + E_n - A_n \leq T + \delta^{\max},
	\end{align}
	that proves part (2).
\end{proof}

\subsection{Proof of Theorem~\ref{thm:gen-node-reordering-time}} \label{sec:appendix_thm_rto}

\begin{proof}
	First, we obtain an upper bound for RTO of the flow separately for each part of the theorem. Second we show that each bound is achievable.
	
	Consider a packet $n$. Let us denote $A_n$ as the arrival time of a packet $n$ (with size $l_n$) into the system and $E_n$ as its exit time. Now, consider another packet $m$ such that $m\leq n-1$ and $E_n < E_m$.
	Since $V$ is a jitter bound for this system.
	\begin{align}
	\left(E_m - A_m\right) - \left(E_n - A_n\right) \leq V.
	\end{align}
	Then:
	\begin{align}
	E_m - E_n  \leq V - \left(A_n- A_m\right).
	\end{align}
	If the flow has arrival curve $\alpha$, then by \eqref{eq:arrcur-bit}, $A_n~-~A_m~\geq~\alpha^{\downarrow}\left(\sum_{k=m}^{n}l_k\right)$:
	\begin{align}
	E_m - E_n  \leq V - \alpha^{\downarrow}\left(\sum_{k=m}^{n}l_k\right).
	\end{align}
	Since $\alpha^{\downarrow}$ is wide-sense increasing and $m\leq n-1$:
	\begin{align}
	\nonumber E_m - E_n  &\leq V - \alpha^{\downarrow}\left(l_{n-1}+l_n\right) \leq V - \alpha^{\downarrow}(2L^{\min}) \\
	&\leq \left[V-\alpha^{\downarrow}(2L^{\min})\right]^+,
	\end{align}
	that  proves item (1) in the statement of theorem.
	
	If the flow has packet-level arrival curve $\alpha_{\mathrm{pkt}}$, then by \eqref{eq:arrcur-pkt}, $A_n - A_m \geq \alpha_{\mathrm{pkt}}^{\downarrow}\left(n-m+1\right)$:
	\begin{align}
	E_m - E_n  \leq V - \alpha_{\mathrm{pkt}}^{\downarrow}(n-m+1).
	\end{align}
	Since $n$ is integer and $n > m$, then $n-m \geq 1$:
	\begin{align}
	E_m - E_n  \leq V - \alpha_{\mathrm{pkt}}^{\downarrow}(2) \leq \left[V - \alpha_{\mathrm{pkt}}^{\downarrow}(2)\right]^+,
	\end{align}
	which proves item (2) in the statement of theorem.
	
	Second, we show that the bounds are achievable by constructing a scenario where the RTO for a packet of the flow reaches the bound in item (1) of the theorem.
	
	Consider two packets $1$ and $2$ with sizes $l_1=l_2=L^{\min}$ and a non order-preserving system with a jitter bound $V$. Packet~$1$ is issued at $A_1=0$ and packet~$2$ arrives at $A_2 = A_1+\alpha^{\downarrow}(l_1+l_2)$, i.e, $(A_2=t_2)$.
	
	Packet~$1$ experiences a delay of $d+V$, $E_1 = A_1+d+V$, and Packet~$2$ experiences a delay of $d$, $E_2 = A_2+d$.
	Then we have:
	\begin{align}
	\nonumber E_1 - E_2 &= A_1+d+V - (A_2+d) = V - (t_2 - t_1) \\
	&= V - \alpha^{\downarrow}(l_1+l_2)= V- \alpha^{\downarrow}(2L^{\min}),
	\end{align}
	which shows that the RTO for packet~$1$ is equal to the bound in part (1) of the theorem.
	
	Now, we verify that jitter bound and arrival curve assumptions are not violated. The difference between the delay of two packets is:
	\begin{align}
	\left(E_1 - A_1\right) - \left(E_2 - A_2\right) = (d+V) - (d) = V,
	\end{align}
	which is equal to the jitter bound. Also:
	\begin{align}
	A_2 - A_1 = t_2 -t_1 = \alpha^{\downarrow}(l_1+l_2),
	\end{align}
	 which shows the arrival curve constraint holds.
	
	The tightness scenario for item (2) of the theorem is similar to the one for part (1). We set $t_2 = t_1 + \alpha^{\downarrow}_{\mathrm{pkt}}(2)$, and the rest follows the description of item (1).
\end{proof}

\subsection{Proof of Theorem~\ref{thm:gen-node-reordering-byte}} \label{sec:appendix_thm_rbo}

\begin{proof}
	
	First, we show the bounds in items (1) and (2).

Item (1). First observe that, from Theorem \ref{thm:gen-node-reordering-time}, if $\alpha(V)<2 L^{\min}$ then there is no reordering and the RBO is $0$.

Next, assume that there is some reordering and  
consider a packet index $m$ such that $\lambda_m>0$. Let $\calE_m = \{i\in \mathbb{Z}^+|i> m, E_i < E_m\}$. Since $\lambda_m>0$ and there is no simultaneous arrival of packets, $\calE_m$ is not empty. Then let $n=\max\{\calE_m\}$.
%
	Due to the jitter bound of this system and since $E_n < E_m$, we have:
	\begin{align}\label{eq:gen-node-reordering-byte-0}
	\nonumber&&\left(E_m - A_m\right) - &\left(E_n - A_n\right) \leq V, \\
&\mbox{ thus } 	&A_n - A_m \leq V +&\left(E_n - E_m \right) <V.
	\end{align}
	Since the flow has arrival curve $\alpha$, then by \eqref{eq:arrcur-bit}, $A_n~-~A_m~\geq~\alpha^{\downarrow}\left(\sum_{k=m}^{n}l_k\right)$. Therefore,
	\begin{align}\label{eq:gen-node-reordering-byte-1}
	\alpha^{\downarrow}\left(\sum_{k=m}^{n}l_k\right) \leq A_n - A_m &< V.
	\end{align}
	By \eqref{eq:gen-node-reordering-byte-11}, we obtain:
	\begin{align}
	\sum_{k=m}^{n}l_k \leq \alpha(V).
	\end{align}
	We exclude packet~$m$ from the left side of the above equation:
	\begin{align}\label{eq:gen-reordering-single-byte-2}
	\sum_{k=m+1}^{n}l_k \leq\alpha(V) - l_m \leq \alpha(V) - L^{\min}.
	\end{align}
	The reordering byte offset $\pi_m$ for packet $m$, defined in \eqref{eq:reordering-byte-def}, includes only the packets with larger index and smaller exit time than packet~$m$. Thus:
	\begin{align}\label{eq:gen-reordering-single-byte-3}
	\pi_m = \sum_{k | k >m, E_k < E_m} l_{k} \leq \sum_{k=m+1}^{n}l_k.
	\end{align}
	Using \eqref{eq:gen-reordering-single-byte-3} in \eqref{eq:gen-reordering-single-byte-2}, we have:
	\begin{align}
	\pi_m \leq  \alpha\left(V\right)- L^{\min},
	\end{align}
	which proves item (1) of the theorem.
	
Item (2). The flow has packet-level arrival curve $\alpha_{\mathrm{pkt}}$. Observe that, from Theorem \ref{thm:gen-node-reordering-time}, if $\alpha_{\mathrm{pkt}}(V)<2$ then there is no reordering and the RBO is $0$.

	Next, since the flow has packet-level arrival curve $\alpha_{\mathrm{pkt}}$, then by \eqref{eq:arrcur-pkt}, $A_n - A_m \geq \alpha_{\mathrm{pkt}}^{\downarrow}\left(n-m+1\right)$. From \eqref{eq:gen-node-reordering-byte-0}, we have:
	\begin{align}\label{eq:gen-node-reordering-byte-2}
	\alpha_{\mathrm{pkt}}^{\downarrow}\left(n-m+1\right) \leq A_n - A_m &< V,
	\end{align}
	By \eqref{eq:gen-node-reordering-byte-11}, we obtain:
	\begin{align}
	n-m+1 \leq  \alpha_{\mathrm{pkt}}(V).
	\end{align}
	Since for any packet $k$, $l_k \leq L^{\max}$:
	\begin{align}
	\sum_{k=m+1}^{n}l_k \leq L^{\max}(n-m) \leq L^{\max} \left(\alpha_{\mathrm{pkt}}(V)-1\right).
	\end{align}
	Since $\pi_m \leq \sum_{k=m+1}^{n}l_k$, item (2) of the theorem is proven.
	
	Second, we show the tightness.
	
	Fix some $\varepsilon>0$, smaller than $\lambda$. By the second technical assumption at the end of Section~\ref{sec:arrcur}, we show the tightness for the two cases, i) $L^{\max}\geq2L^{\min}$, ii) $L^{\min}= L^{\max}$.

	Case i) $L^{\max}\geq2L^{\min}$. By assumption, we know that  $\alpha(0^+)\geq L^{\max}$; therefore, $\alpha(0^+)\geq 2L^{\min}$.
	
	By the second technical assumption at the end of Section~\ref{sec:arrcur}, there exists an integer $n$ and a sequence of packet lengths $l_k\in [L^{\min}, L^{\max}]$ such that $l_1 = L^{\min}$ and $\sum_{k=2}^n l_k = \alpha(V-\varepsilon) - L^{\min}$. Since $\varepsilon <\lambda$ and by Theorem \ref{thm:gen-node-reordering-time} $\lambda \leq V$, $\alpha(V-\varepsilon) \geq  2L^{\min}$; therefore, $n\geq 2$. Now, since the arrival curve is achievable, there also exists a sequence of emission times $A_1=0, ...A_n=V- \varepsilon$ such that the packet sequence $A_1, ...A_n ,l_1, ..., l_n$ satisfies the arrival curve constraint $\alpha$.
	
	
	
	Next, we construct the exit times of packets $k$ from the system as follows:
	\begin{align}
	 E_1=V+\varepsilon, \indent E_k=V+\frac{(k-2)\varepsilon}{n}, ~~k=2,\dots,n.
	\end{align}
	Observe that $E_2<E_3<\dots<E_n<E_1$. Also note that $A_1\leq A_2\leq \dots \leq A_n=V-\varepsilon < E_2 = V$.
	


	Now, according to \eqref{eq:reordering-byte-def}, the RBOs for packet~$1$ and packet~$k$, $k=2,\dots,n$, are:
	\begin{align}
	\nonumber \pi_1 &= \sum_{j | j >1, E_j < E_1} l_{j} = \sum_{k=2}^n l_k = \alpha(V-\varepsilon) - L^{\min},\\
	\pi_k &=\sum_{j | j >k, E_j < E_k} l_{j}= 0.
	\end{align}
	Therefore, $\pi = \max_{1\leq i\leq n}\{\pi_i\} = \alpha(V-\varepsilon) - L^{\min}$.
	
	Finally, we verify that the assumptions are not violated: 1)  arrival curve, 2) jitter bound, 3) RTO of the flow.
	
	 (1) The arrival curve constraint is satisfied by construction.
	
	
	(2) For any packet $k\geq 1$, we have:
		\begin{align}
		E_k - A_k &\leq E_1 - A_1 = V+\varepsilon,\\
		E_k - A_k &\geq E_2 - A_n = V - (V-\varepsilon) = \varepsilon.
		\end{align}
		Therefore, the jitter is:
		\begin{align}
		\max_{k}\{E_k - A_k\} - \min_{k}\{E_k - A_k\} \leq V+\varepsilon - \varepsilon = V,
		\end{align}
		which conforms the jitter constraint.

	(3) For any packet $k\geq 2$, the packets are in order
	\begin{align}
	\lambda_k = E_k - \min_{j | j \geq k, E_j \leq E_k} = E_k - E_k = 0.
	\end{align}
	For packet $1$, we have:
	\begin{align}
	\nonumber\lambda_1 &= E_1 - \min_{j | j \geq 1, E_1 \leq E_k} = E_1 - E_2 \\
	&= (V+\varepsilon) - V = \varepsilon,
	\end{align}
	therefore, $\lambda = \max_{1\leq i\leq n}\{\lambda_i\} = \varepsilon$, that satisfies the RTO constraint.

Thus, we have shown that, for every $\varepsilon$ small enough, there is a scenario where the RBO reaches $\alpha(V-\varepsilon)-L^{\min}$. Thus the minimal bound is at least $\sup_{\varepsilon>0}\alpha(V-\varepsilon)-L^{\min}=\alpha(V)-L^{\min}$ because arrival curves are left-continuous.

	Case ii) $L^{\max}=L^{\min}$. Then all the packets have the same size $l=L^{\max}$. By assumption, we know that  $\alpha(0^+)\geq l$.
	
	By the second technical assumption at the end of Section~\ref{sec:arrcur}, there exists an integer $n$ and a sequence of packets with length $l$ such that $\sum_{k=1}^n l_k = \alpha(V-\varepsilon) - l$. Now, if $\alpha(V-\varepsilon)<2l$, then $n=1$ and therefore no reordering occurs and $\pi=0$. Hence, we consider the case $\alpha(V-\varepsilon)\geq 2l$; then $n\geq 2$. Now, since the arrival curve is achievable, there also exists a sequence of emission times $t_1=0, ...t_n=V- \varepsilon$ such that the packet sequence $t_1, ...t_n ,l_1, ..., l_n$ satisfies the arrival curve constraint $\alpha$. The rest of the proof follows exactly as case (i).

	The tightness scenario for item (2) of the theorem is similar to the one for case (ii). We set $n = \alpha_{\mathrm{pkt}}(V)$ and for any $k=1,\dots,n$, $l_k = L^{\max}$. 
\end{proof}

\subsection{Proof of Theorem~\ref{thm:gen-concat-reordering-time}}\label{sec:appendix_thm_rto-concat}
\begin{figure}[h]
	\centering
	\includegraphics[width=0.8 \linewidth]{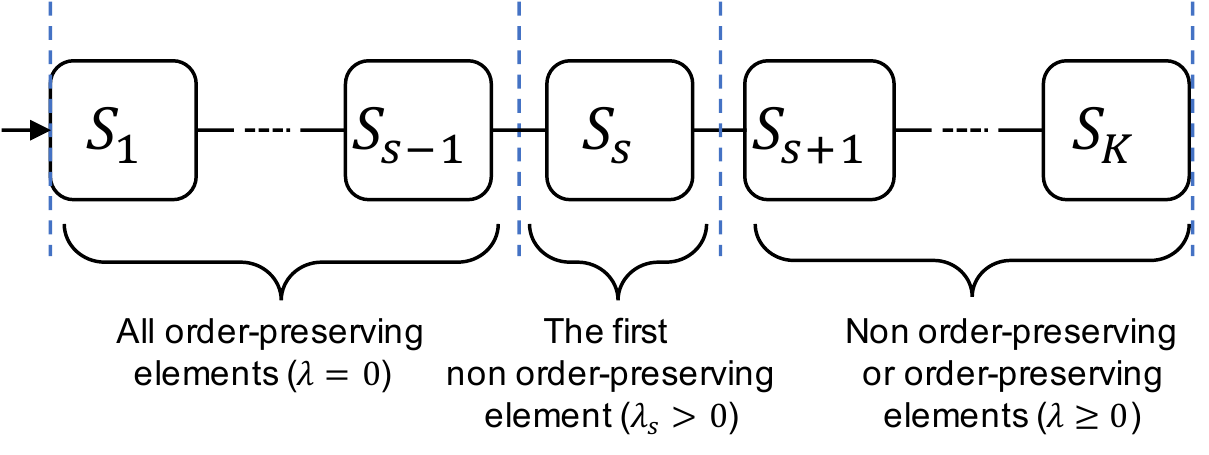}
	\caption{Notation for the sequence of network elements used in Theorem \ref{thm:gen-concat-reordering-time}.}
	\label{fig:gen-k-node}
\end{figure}
\begin{proof}
		First, we obtain an upper bound for RTO of the flow separately for each part of the theorem. Second we show that each bound is achievable.
		
	Consider Fig. \ref{fig:gen-k-node}. We denote the arrival and exit times of a packet~$i$ at $S_h$ by $E^{h-1}_i$ and $E^{h}_i$. Consider two packet indices $m$ and $n$ such that $m<n$ and $E^{K}_n < E^{K}_m$.
	 Since $S_s$ is the first system with nonzero RTO, for any system $S_h$, $h<s$, $E^{h}_m < E^{h}_n$.
	Therefore, $E^{s-1}_{m} \leq E^{s-1}_n$, i.e., at the output of system $S_{s-1}$ (input of system $S_s$), packet~$m$ and $n$ are in-order. Therefore, by definition of the RTO bound at $S_s$,
\begin{equation}\label{eq-pr7a}
E^{s}_m~-~E^{s}_n~\leq~\lambda_s.
\end{equation}

	Now, according to the jitter bound for the concatenation of systems $S_{s+1}$ to $S_K$, we have:
	\begin{align}
	\left(E^{K}_m - E^{s}_m\right) - \left(E^{K}_n - E^{s}_n\right) \leq \sum_{h=s+1}^{K} V_h.
	\end{align}
	Then, we have:
	\begin{align}\label{eq:gen-reordering}
	E^{K}_m - E^{K}_n \leq \left( E^{s}_m- E^{s}_n\right)+ \sum_{h=s+1}^{K} V_h.
	\end{align}
	Combining with \eqref{eq-pr7a}:
	\begin{align}\label{eq:proof-gen-reordering-1}
	E^{K}_m - E^{K}_n \leq \lambda_s+ \sum_{h=s+1}^{K} V_h \eqdef \Lambda(K).
	\end{align}
	
	Second, we show tightness. We are given a sequence of systems with RTOs $\lambda_h$ and jitters $V_h$, and we construct a scenario that conforms with these parameters and where a packet reaches the RTO bound in Theorem \ref{thm:gen-concat-reordering-time} at system $K$.
	We use the same notation as before. In particular, $S_s$ is the first system in the sequence for which $\lambda_s>0$.
		Now, consider a trace with two packets $1$ and $2$ entering $S_1$. Also, consider a 
set of positive values $\{d_s,d_{s+1},\dots,d_K\}$.
	
	Packet $1$ and $2$ arrive at $S_1$ at times $E_1=0$ and $E_2=\varepsilon$, with $\lambda_s>\varepsilon>0$. Each packet then has the same transfer time through system $S_j$ at time $t_j$, $j\in\{1,2,\dots,s-1\}$, thus preserving order and $E^{s-1}_2 = E^{s-1}_1+\varepsilon$.
	
	 The transfer times through $S_s$ are $d_s+\lambda_s$ for packet~$1$  and $d_s$ for packet $2$, i.e.,  $E^s_1 = E^{s-1}_1 + d_s+\lambda_s$ and $E^s_2 =E^{s-1}_2+d_s$. Packet~$1$ and $2$ experience the delays of $d_h+V_h$ and $d_h$ at system $S_h$, $h\in\{s+1,s+2,\dots,K\}$:
	\begin{align}
	\nonumber E^{h}_1 &= E^{h-1}_1 + d_h+ V_h,\\
	E^{h}_2 &= E^{h-1}_2 + d_h.
	\end{align}
	Then at the output of system $S_K$, we have:
	\begin{align}
	\nonumber E^{K}_1 &= E^{s}_1 + \sum_{h=s+1}^{K} \left(d_h+ V_h\right) =E^{s-1}_1+\lambda_s+ \sum_{h=s}^{K} d_h + \sum_{h=s+1}^{K} V_h,\\
	E^{K}_2 &= E^{s}_2 + \sum_{h=s+1}^{K} d_h = E^{s-1}_1+ \varepsilon+\sum_{h=s}^{K} d_h.
	\end{align}

	We now verify that the assumptions in the statement of theorem are not violated.
	
	(1) We check that RTOs for all the systems are not violated. For any system $S_j, j<s$, packets $1$ and $2$ preserve order by construction, i.e. $\lambda_j=0$.

	For system $S_s$, according to the departure times of packets $1$ and $2$ from $S_s$, we have:
	\begin{align}
	E^s_1 - E^s_2 =  \lambda^s - \varepsilon,
	\end{align}
	which satisfies the constraint $\lambda_s$ on RTO conforms to RTO assumption for system $S_s$ being the first system that RTO is equal to $\lambda_s$. At the output of system $S_s$ we have $E^s_2 < E^s_2$, i.e., packet~$2$ is prior to packet~$1$. For system $S_h$, $h\in\{s+1,s+1,\dots,K\}$, we have
	\begin{align}
	\nonumber E^h_2 - E^h_1 &= \left(E^{s-1}_1+\varepsilon+ \sum_{j=s}^{h} d_j\right) \\
	\nonumber&- \left(E^{s-1}_1+\lambda_s+ \sum_{j=s}^{h} d_j + \sum_{j=s+1}^{h} V_j\right)\\
	&=\varepsilon-\lambda_s-\sum_{j=s+1}^{h} V_j < 0,
	\end{align}
	that shows $E^h_2 < E^h_1$, i.e., system $S_h,  h\geq s+1$ preserves the order of its input. Therefore, RTO for each system $S_h$ is $0$, which satisfies any RTO constraint.
	
	(2) We check that the jitter bounds are not violated.
	
	For systems $S_1$ to $S_{s-1}$, since both packets experience the same delay, the jitter is $0$.
	
	For system $S_s$, we have:
	\begin{align}
	\nonumber \left(E^s_1 - E^{s-1}_1\right) - \left(E^s_2 - E^{s-1}_2\right) &= \lambda_s.
	\end{align}
	Now, by Theorem \ref{thm:gen-node-reordering-time}, $\lambda_s \leq V_s$, thus the jitter bound assumption for $S_s$ are satisfied.
	
	For any system $S_h$ where $h \in \{s+1,\dots,K\}$, we have:
	\begin{align}
	\nonumber \left(E^h_1 - E^{h-1}_1\right) - \left(E^h_2 - E^{h-1}_2\right) = (d_h+V_h) - (d_h) = V_h,
	\end{align}
	which shows that the jitter bound assumptions for systems $S_{s+1}, S_{s+2}, \dots, S_K$ are satisfied.

Last,
	\begin{align}
	E^K_1 - E^K_2 = \lambda_s -\varepsilon + \sum_{h=s+1}^{K} V_h,
	\end{align}
	thus the RTO for packet~$1$ is equal to the bound in Theorem \ref{thm:gen-concat-reordering-time} minus $\varepsilon$. Since $\varepsilon$ can be arbitrarily small, this shows the result.

\end{proof}

\section{Details of Computations for Case Study 1 in Section \ref{subsec:case1}}\label{sec:appendix-numerical-case1}
\begin{table}[h]
	\caption{Arrival curve propagation from point 1 to point 9 under lossless network condition. Arrival curve at point $i$ is $\alpha_i(t)=\min(6.4e3~t+b_i,125e6~t+M_i)$, where $b_i$ and $M_i$ are in bytes and shown in the table. Since under lossless network condition, the re-sequencing buffers do not increase  delay and jitter bounds, the arrival curves are the same for the four placement strategies.}
	\centering
	\begin{tabular}{|c|c|c|c|c|c|c|c|c|c|}
		\hline
		Burst & 1 & 2 & 3 & 4 & 5 & 6 & 7 & 8 & 9 \\ \hline
		$b_i$ & $6400$ & \multicolumn{2}{c|}{$6400$} & \multicolumn{2}{c|}{$6400$} & \multicolumn{2}{c|}{$6400$} & \multicolumn{2}{c|}{$6400$} \\ \hline
		$M_i$ & $64$ & \multicolumn{2}{c|}{$251$} & $1751$ & $64$ & \multicolumn{2}{c|}{$251$} & $1751$ & $64$ \\ \hline
	\end{tabular}
\label{table:case1-arr-lossless}
\end{table}

\begin{table*}[t]
	\caption{Arrival curve propagation from point 1 to point 9 under lossy network condition. Arrival curve at point $i$ is $\alpha_i(t)=\min(6.4e3~t+b_i,125e6~t+M_i)$, where $b_i$ and $M_i$ are in bytes and shown in the table.}
	\centering
	\resizebox{0.8\textwidth}{!}{
	\begin{tabular}{|c|c|c|c|c|c|c|c|c|c|c|}
		\hline
		Re-sequencing & Burst & 1 & 2 & 3 & 4 & 5 & 6 & 7 & 8 & 9 \\ \hline
		\multirow{2}{*}{Only at $h_2$} & $b_i$ & $6400$ & \multicolumn{2}{c|}{$6400$} & \multicolumn{2}{c|}{$6400$} & \multicolumn{2}{c|}{$6400$} & \multicolumn{2}{c|}{$6400$} \\ \cline{2-11}
		& $M_i$ & $64$ & \multicolumn{2}{c|}{$251$} & $1751$ & $64$ & \multicolumn{2}{c|}{$251$} & $1751$ & $64$ \\ \hline
		\multirow{2}{*}{Only at $S_2$} & $b_i$ & $6400$ & \multicolumn{2}{c|}{$6400$} & \multicolumn{2}{c|}{$6400$} & $6400$ & $6400$ & \multicolumn{2}{c|}{$6400$} \\ \cline{2-11}
		& $M_i$ & $64$ & \multicolumn{2}{c|}{$251$} & $1751$ & $64$ & $251$ & $2249$ & $3749$ & $64$ \\ \hline
		\multirow{2}{*}{At $S_1$ and $h_2$} & $b_i$ & $6400$ & $6400$ & $6400$ & \multicolumn{2}{c|}{$6400$} & \multicolumn{2}{c|}{$6400$} & \multicolumn{2}{c|}{$6400$} \\ \cline{2-11}
		& $M_i$ & $64$ & $251$ & $626$ & $1875$ & $64$ & \multicolumn{2}{c|}{$2012$} & $1751$ & $64$ \\ \hline
		\multirow{2}{*}{At $S_1$ and $S_2$} & $b_i$ & $6400$ & $6400$ & $6400$ & \multicolumn{2}{c|}{$6400$} & $6400$ & $6400$ & \multicolumn{2}{c|}{$6400$} \\ \cline{2-11}
		& $M_i$ & $64$ & $251$ & $626$ & $1875$ & $64$ & $251$ & $375$ & $1875$ & $64$ \\ \hline
	\end{tabular}
}
\label{table:case1-arr-lossy}
\end{table*}

In these scenarios, we require arrival curve information to compute RTO, RBO, and delay upper bounds. Let us call $\alpha_i$ as arrival curve of flow $f$ at point $i$ in units of bytes. Tables \ref{table:case1-arr-lossless} and \ref{table:case1-arr-lossy} show the propagated arrival curve at points 1 to 9 in Fig. \ref{fig:auto}. The arrival curve at point $i$ has the form of $\alpha_i(t)=\min(r~t+b_i,c~t+M_i)$, where $r=6400$ bytes per second, $c=125e5$ bytes per second~(i.e. $1Gbps$), $b_i$ and $M_i$ are shown in Tables \ref{table:case1-arr-lossless} and \ref{table:case1-arr-lossy}. $\alpha_0(t) = r~t+b_0$, where $b_0=6400~bytes$. The arrival curves at point 1, 5, and 9 capture the line shaping and packetizer effects:
\begin{align}
\nonumber\alpha_1(t)&=\min(r~t+b_0,c~t+512),\\
\nonumber \alpha_5(t)&=\min(r~t+b_4,c~t+512),\\
\alpha_9(t)&=\min(r~t+b_8,c~t+512).
\end{align}
The arrival curves at points 4 and 8 capture the effect of traversing the output FIFO systems with service curve $\beta(t)=125e6[t-Q]^+~bytes$, where $Q=12\mu s$:
\begin{align}
\nonumber\alpha_4(t)&=\min(r~t+b_3+rQ,c~t+M_3+cQ),\\
 \alpha_8(t)&=\min(r~t+b_7+rQ,c~t+M_7+cQ).
\end{align}
The arrival curves at point 2 and 6 capture the effect of switching fabric with jitter bound $V_{\mathrm{sf}}$:
\begin{align}
\nonumber \alpha_2(t)&=\min(r~t+b_1+rV_{\mathrm{sf}},c~t+M_1+cV_{\mathrm{sf}}), \\
 \alpha_6(t)&=\min(r~t+b_5+rV_{\mathrm{sf}},c~t+M_5+cV_{\mathrm{sf}}).
\end{align}
The arrival curves at point 3 and 7 capture the effect of re-sequencing buffer (if any) at switches $S_1$ and $S_2$ with time-out values respectively $T_1$ and $T_2$ under lossy network condition:
\begin{align}
\nonumber\alpha_3(t)&=\min(r~t+b_2+rT_1,c~t+M_2+cT_1), \\
\alpha_7(t)&=\min(r~t+b_6+rT_2,c~t+M_6+cT_2).
\end{align}
Under lossless network condition, the re-sequencing buffers do not increase delay and jitter bounds, hence, $\alpha_3=\alpha_2$ and $\alpha_7=\alpha_6$.

To compute delay bounds of output FIFO systems we use the results in \cite{mohammadpour_improved_2019-1, le_boudec_network_2001}. Accordingly, the delay bound of FIFO system of $h_1$ is $\delta^{\max}_{\mathrm{FIFO},h_1} = 63.2\mu s$. Also, the minimum delay for each FIFO system is the transmission of a packet with minimum length, $\delta^{\min}_{\mathrm{FIFO},h_1} = \delta^{\min}_{\mathrm{FIFO},S_1}  =\delta^{\min}_{\mathrm{FIFO},S_2} = \frac{L^{\min}}{c} = 512 ns$. Now, since we know delay upper and lower bounds for the FIFO system of $h_1$, its jitter is $V_{\mathrm{FIFO},h_1} = 62.69\mu s$. Now we analyze the four strategies separately.
\subsubsection{Re-sequencing only at $h_2$}
We first obtain delay and jitter bounds of flow $f$ for the FIFO output ports of $S_1$ and $S_2$. Using the arrival curves in Tables \ref{table:case1-arr-lossless} and \ref{table:case1-arr-lossy}:
\begin{align}
\delta^{\max}_{\mathrm{FIFO},S_1} = \delta^{\max}_{\mathrm{FIFO},S_2} = 14.01\mu s.
\end{align}
Since, we already compute the minimum delay for the FIFO systems of $S_1$ and $S_2$, the jitter bounds are:
\begin{align}
V_{\mathrm{FIFO},S_1} = V_{\mathrm{FIFO},S_2} = 13.5\mu s.
\end{align}
Having the knowledge on jitters of each element, we now compute the RTO bound of the flow at $h_2$. To obtain the RTO bound, we use Theorem \ref{thm:gen-concat-reordering-time}. Since the switching fabric (SF) in $S_1$ is the first non order-preserving element, we need to compute its corresponding RTO. Using Theorem \ref{thm:gen-node-reordering-time}:
\begin{align}
\nonumber \lambda_{\mathrm{SF},S_1} = \left[V_{\mathrm{SF},S_1} - \alpha_1^{\downarrow}(2L^{\min})\right]^+ &= 1.5 \mu s- 0.512 \mu s \\
&=  0.988 \mu s.
\end{align}
Therefore the RTO at $h_2$ is:
\begin{align}
\nonumber \Lambda(h_2) &= \lambda_{\mathrm{SF},S_1} + V_{\mathrm{FIFO},S_1}+V_{\mathrm{SF},S_2} +V_{\mathrm{FIFO},S_2} \\
&= 0.988 \mu + 13.5 \mu + 1.5 \mu + 13.5 \mu =29.49 \mu s.
\end{align}
Then $T_{h_2}=\Lambda(h_2) = 29.49 \mu s$.

Now, if the network is lossless, due to Theorem \ref{thm:delay-reseq}, re-sequencing is for free; therefore, the bounds are:
\begin{align}
\nonumber \delta^{0,\max}_{\mathrm{e2e}} &= \delta^{\max}_{\mathrm{FIFO},h_1} + \delta^{\max}_{\mathrm{SF},S_1} + \delta^{\max}_{\mathrm{FIFO},S_1} + \delta^{\max}_{\mathrm{SF},S_2} + \delta^{\max}_{\mathrm{FIFO},S_2} \\
\nonumber &= 63.2 + 2 + 14.01 + 2 + 14.01 = 95.22 \mu s,\\
\nonumber V^0_{\mathrm{e2e}} &= V_{\mathrm{FIFO},h_1} + V_{\mathrm{SF},S_1} + V_{\mathrm{FIFO},S_1} + V_{\mathrm{SF},S_2} + V_{\mathrm{FIFO},S_2} \\
\nonumber&= 62.69 + 1.5+13.5+1.5+13.5 = 92.69 \mu s.
\end{align}
Using Corollary \ref{col:reordering-concat}, the RBO bound at $h_2$ is:
\begin{align}
\nonumber \Pi(h_2) &= \alpha(V_{\mathrm{FIFO},h_1} + V_{\mathrm{SF},S_1} + V_{\mathrm{FIFO},S_1} + V_{\mathrm{SF},S_2}) - L^{\min}\\
&= \alpha(79.19\mu s) - 64= 6336~bytes.
\end{align}
Then $B_{h_2}=\Pi(h_2)=6336~bytes$. Note that due to Theorem \ref{thm:gen-node-reordering-byte}, we eliminate the FIFO system at $S_2$ as the switching fabric of $S_2$ is the last FIFO system.

If the network is lossy, by Theorem \ref{thm:delay-reseq}, the jitter and delay worst-case are increased by $T_{h_2} = 29.49\mu s$; therefore:
\begin{align}
\delta^{\max}_{\mathrm{e2e}} = 124.72 \mu s, \indent\indent V_{\mathrm{e2e}} = 122.19 \mu s.
\end{align}
The size of re-sequencing buffer is:
\begin{align}
\nonumber B_{h_2}&=\alpha(V_{\mathrm{FIFO},h_1} + V_{\mathrm{SF},S_1} + V_{\mathrm{FIFO},S_1} + V_{\mathrm{SF},S_2}\\
\nonumber&+V_{\mathrm{FIFO},S_2} + T_{h_2})= \alpha(122.19 \mu s)=6400~bytes.
\end{align}

\subsubsection{Re-sequencing only at $S_2$}
Similarly to the previous scenario, $\delta^{\max}_{\mathrm{FIFO},S_1} = 14.01 \mu s$ and $V_{\mathrm{FIFO},S_1} = 13.5\mu s$; and therefore, we obtain RTO bound in switching fabric of $S_1$ as $\lambda_{\mathrm{SF},S_1} = 0.988\mu s$. Using Theorem \ref{thm:gen-concat-reordering-time}, the RTO bound after the switching fabric of $S_2$ is:
\begin{align}
\nonumber \Lambda(S_2) &= \lambda_{\mathrm{SF},S_1} + V_{\mathrm{FIFO},S_1}+V_{\mathrm{SF},S_2} \\
&= 0.988 \mu + 13.5 \mu + 1.5 \mu =15.99 \mu s.
\end{align}
Then $T_{S_2}=\Lambda(S_2) = 15.99 \mu s$.

Now, if the network is lossless, $\delta^{\max}_{\mathrm{FIFO},S_2} = 14.01 \mu s$. Then, due to Theorem \ref{thm:delay-reseq}, re-sequencing is for free; therefore, similarly to the previous strategy $\delta^{0,\max}_{\mathrm{e2e}} = 95.22 \mu s$ and $V^0_{\mathrm{e2e}} = 92.69 \mu s$. Also for the RBO bound at $S_2$, using Corollary \ref{col:reordering-concat}, we have:
\begin{align}
\nonumber\Pi(S_2) &= \alpha(V_{\mathrm{FIFO},h_1} + V_{\mathrm{SF},S_1} + V_{\mathrm{FIFO},S_1} + V_{\mathrm{SF},S_2}) - L^{\min}\\
&= \alpha(79.19\mu s) - 64= 6336~bytes.
\end{align}
By Theorem \ref{thm:reseq-param},  and the size is $B_{S_2}=\Pi(S_2)=6336~bytes$.

If the network is lossy, by Theorem \ref{thm:delay-reseq}, the jitter and delay worst-case are increased by $T_{S_2} = 15.99\mu s$; therefore this affects the arrival curve at point 7 and in turn delay bound of output FIFO system at $S_2$. Then, $\delta^{\max}_{\mathrm{FIFO},S_2} = 30 \mu s$ and $V_{\mathrm{FIFO},S_2} = 29.49\mu s$. Finally,
\begin{align}
\nonumber \delta^{\max}_{\mathrm{e2e}} = \delta^{\max}_{\mathrm{FIFO},h_1}& + \delta^{\max}_{\mathrm{SF},S_1} + \delta^{\max}_{\mathrm{FIFO},S_1} + \delta^{\max}_{\mathrm{SF},S_2} + \delta^{\max}_{\mathrm{FIFO},S_2}\\
\nonumber &+ T_{S_2}
= 127.22 \mu s,\\
\nonumber V_{\mathrm{e2e}} = V_{\mathrm{FIFO},h_1} &+ V_{\mathrm{SF},S_1} + V_{\mathrm{FIFO},S_1} + V_{\mathrm{SF},S_2} + V_{\mathrm{FIFO},S_2}\\
\nonumber&+ T_{S_2}
= 124.69 \mu s.
\end{align}
The size of re-sequencing buffer is:
\begin{align}
\nonumber B_{S_2}&=\alpha(V_{\mathrm{FIFO},h_1} + V_{\mathrm{SF},S_1} + V_{\mathrm{FIFO},S_1} + V_{\mathrm{SF},S_2}+ T_{S_2})\\
&= \alpha(95.18 \mu s)=6400~bytes.
\end{align}
\subsubsection{Re-sequencing at $S_1$ and $h_2$}
Since we already computed arrival curve at point 1, we compute RTO bound switching fabric at $S_1$ using Theorems \ref{thm:gen-node-reordering-time}, $\lambda_{\mathrm{SF},S_1} = 0.988 \mu s$. Then $T_{S_1}=\lambda_{\mathrm{SF},S_1} = 0.988 \mu s$. Similarly to switching fabric of $S_1$, we have $\lambda_{\mathrm{SF},S_2} = 0.988 \mu s$.

Now, if the network is lossless, using the arrival curves in Table \ref{table:case1-arr-lossless}, $\delta^{\max}_{\mathrm{FIFO},S_1} = \delta^{\max}_{\mathrm{FIFO},S_2} = 14.01\mu s$ and $V_{\mathrm{FIFO},S_1} = V_{\mathrm{FIFO},S_2} = 13.5\mu s$. Hence, due to Theorem \ref{thm:delay-reseq}, re-sequencing is for free; therefore, similarly to the previous strategy $\delta^{0,\max}_{\mathrm{e2e}} = 95.22 \mu s$ and $V^0_{\mathrm{e2e}} = 92.69 \mu s$. To compute RTO bound at $h_2$, we need to find RTO bound switching fabric of $S_2$ as the first non order-preserving element after re-sequencing buffer of $S_1$ (the output of which is in-order). Using Theorem \ref{thm:gen-concat-reordering-time}, the RTO bound at $h_2$ is:
\begin{align}
\nonumber \Lambda(h_2) &= \lambda_{\mathrm{SF},S_2} +V_{\mathrm{FIFO},S_2} = 0.988 \mu + 13.5 \mu =14.49 \mu s.
\end{align}
Then $T_{h_2}=\Lambda(h_2) = 14.49 \mu s$. Also for the RBO bound at $S_1$ and $h_1$, using Corollary \ref{col:reordering-concat}, we have:
\begin{align}
\nonumber\Pi(S_1) &= \alpha(V_{\mathrm{FIFO},h_1} + V_{\mathrm{SF},S_1}) - L^{\min} \\
\nonumber&= \alpha(64.19\mu s) - 64= 6336~bytes,\\
\nonumber\Pi(h_2) &= \alpha(V_{\mathrm{FIFO},h_1} + V_{\mathrm{SF},S_1} + V_{\mathrm{FIFO},S_1} + V_{\mathrm{SF},S_2}) - L^{\min}\\
\nonumber&= \alpha(79.19\mu s) - 64= 6336~bytes.
\end{align}
By Theorem \ref{thm:reseq-param}, $B_{S_1} = B_{h_2}=\Pi(S_2)=6336~bytes$.

If the network is lossy, the re-sequencing buffer at $S_1$ increases the jitter by $T_{S_1}$; the impact on arrival curves is shown in Table \ref{table:case1-arr-lossy}. Hence, $\delta^{\max}_{\mathrm{FIFO},S_1} = 15\mu s, \delta^{\max}_{\mathrm{FIFO},S_2} = 14.01\mu s$. Using the obtained lower and upper delay bounds, $V_{\mathrm{FIFO},S_1} = 14.49 \mu s,  V_{\mathrm{FIFO},S_2} = 13.5\mu s$. To compute RTO bound at $h_2$, similarly to the lossless case,
\begin{align}
\nonumber\Lambda(h_2) &= \lambda_{\mathrm{SF},S_2} +V_{\mathrm{FIFO},S_2} = 0.988 \mu + 13.5 \mu =14.49 \mu s.
\end{align}
Then $T_{h_2}=\Lambda(h_2) = 14.49 \mu s$. Finally:
\begin{align}
\nonumber \delta^{\max}_{\mathrm{e2e}} = \delta^{\max}_{\mathrm{FIFO},h_1} &+ \delta^{\max}_{\mathrm{SF},S_1} + \delta^{\max}_{\mathrm{FIFO},S_1} + \delta^{\max}_{\mathrm{SF},S_2} + \delta^{\max}_{\mathrm{FIFO},S_2}\\
\nonumber&+ T_{S_1} + T_{h_2} =
111.72 \mu s,\\
\nonumber V_{\mathrm{e2e}} = V_{\mathrm{FIFO},h_1} &+ V_{\mathrm{SF},S_1} + V_{\mathrm{FIFO},S_1} + V_{\mathrm{SF},S_2} + V_{\mathrm{FIFO},S_2}\\
\nonumber&+ T_{S_1}+ T_{h_2} =
109.19 \mu s.
\end{align}
The size of re-sequencing buffers are:
\begin{align}
\nonumber B_{S_1}=\alpha(V_{\mathrm{FIFO},h_1} &+V_{\mathrm{SF},S_1}+  T_{S_1})= \alpha(65.18 \mu s)\\
\nonumber&=6400~bytes,\\
\nonumber B_{h_2}=\alpha(V_{\mathrm{FIFO},h_1} &+V_{\mathrm{SF},S_1}+  T_{S_1} + V_{\mathrm{FIFO},S_1} + V_{\mathrm{SF},S_2}\\
\nonumber&+ V_{\mathrm{FIFO},S_2} + T_{h_2})= \alpha(109.19 \mu s)\\
\nonumber&=6400~bytes.
\end{align}

\subsubsection{Re-sequencing at $S_1$ and $S_2$}
Similarly to the previous scenario, we have $T_{S_1}=\lambda_{\mathrm{SF},S_1} = 0.988\mu s$. Also $\lambda_{\mathrm{SF},S_2} = 0.988 \mu s$.

Now, if the network is lossless, using the arrival curves in Table \ref{table:case1-arr-lossless}, $\delta^{\max}_{\mathrm{FIFO},S_1} = \delta^{\max}_{\mathrm{FIFO},S_2} = 14.01\mu s$ and $V_{\mathrm{FIFO},S_1} = V_{\mathrm{FIFO},S_2} = 13.5\mu s$. Hence, due to Theorem \ref{thm:delay-reseq}, re-sequencing is for free; therefore, similarly to the previous strategy $\delta^{0,\max}_{\mathrm{e2e}} = 95.22 \mu s$ and $V^0_{\mathrm{e2e}} = 92.69 \mu s$. To compute RTO bound at $h_2$, we need to find RTO bound switching fabric of $S_2$ as the first non order-preserving element after re-sequencing buffer of $S_1$ (the output of which is in-order). Using Theorem \ref{thm:gen-concat-reordering-time}, the RTO bound at $S_2$ is $\Lambda(S_2) = \lambda_{S_2} = 0.988\mu s$.  Also for the RBO bound at $S_1$ and $S_2$, using Corollary \ref{col:reordering-concat}, we have:
\begin{align}
\nonumber \Pi(S_1) &= \alpha(V_{\mathrm{FIFO},h_1} + V_{\mathrm{SF},S_1}) - L^{\min} \\
\nonumber &= \alpha(64.19\mu s) - 64= 6336~bytes,\\
\nonumber \Pi(S_2) &= \alpha(V_{\mathrm{FIFO},h_1} + V_{\mathrm{SF},S_1} + V_{\mathrm{FIFO},S_1} + V_{\mathrm{SF},S_2}) - L^{\min}\\
\nonumber &= \alpha(79.19\mu s) - 64= 6336~bytes.
\end{align}
By Theorem \ref{thm:reseq-param}, $B_{S_1} = B_{S_2}=\Pi(S_2)=6336~bytes$.

If the network is lossy, the re-sequencing buffer at $S_1$ increases the jitter by $T_{S_1}$; the impact on arrival curves is shown in Table \ref{table:case1-arr-lossy}. Hence, $\delta^{\max}_{\mathrm{FIFO},S_1} = 15\mu s, \delta^{\max}_{\mathrm{FIFO},S_2} = 14.01\mu s$. Using the obtained lower and upper delay bounds, $V_{\mathrm{FIFO},S_1} = 14.49 \mu s,  V_{\mathrm{FIFO},S_2} = 13.5\mu s$. To compute RTO bound at $h_2$, similarly to the lossless case, $\Lambda(S_2) = \lambda_{\mathrm{SF},S_2} = 0.988 \mu s$. Then $T_{S_2}=\Lambda(S_2) = 0.988 \mu s$. Finally:
\begin{align}
\nonumber \delta^{\max}_{\mathrm{e2e}} = \delta^{\max}_{\mathrm{FIFO},h_1} &+ \delta^{\max}_{\mathrm{SF},S_1} + \delta^{\max}_{\mathrm{FIFO},S_1} + \delta^{\max}_{\mathrm{SF},S_2} + \delta^{\max}_{\mathrm{FIFO},S_2}\\
\nonumber &+ T_{S_1} + T_{S_2} =
99.22 \mu s,\\
\nonumber V_{\mathrm{e2e}} = V_{\mathrm{FIFO},h_1} &+ V_{\mathrm{SF},S_1} + V_{\mathrm{FIFO},S_1} + V_{\mathrm{SF},S_2} + V_{\mathrm{FIFO},S_2}\\
\nonumber &+ T_{S_1}+ T_{S_2} =
96.69 \mu s.
\end{align}
The size of re-sequencing buffers are:
\begin{align}
\nonumber B_{S_1}=\alpha(V_{\mathrm{FIFO},h_1} &+V_{\mathrm{SF},S_1}+  T_{S_1})= \alpha(65.18 \mu s)\\
\nonumber &=6400~bytes,\\
\nonumber B_{S_2}=\alpha(V_{\mathrm{FIFO},h_1} &+V_{\mathrm{SF},S_1}+  T_{S_1} + V_{\mathrm{FIFO},S_1} + V_{\mathrm{SF},S_2}\\
\nonumber &+ T_{h_2})= \alpha(83.19 \mu s)=6400~bytes.
\end{align}

\end{document}